\documentclass[letterpaper]{article} 
\usepackage[square,numbers]{natbib}
\usepackage{authblk}
\usepackage[margin=1in]{geometry} 

\usepackage{amsmath}
\usepackage{amssymb}
\usepackage{amsthm}
\usepackage{mathtools}
\usepackage{optidef}
\usepackage{listings}
\usepackage{fullpage}
\usepackage{pifont}
\usepackage[utf8]{inputenc}
\usepackage[shortlabels]{enumitem}
\usepackage[colorlinks=true, urlcolor=blue]{hyperref}
\usepackage{cleveref}
\usepackage{xspace}
\usepackage{booktabs}
\usepackage{algorithm}
\usepackage{wasysym}
\usepackage[noend]{algpseudocode}
\usepackage{tikz}
\usepackage{enumitem}
\usepackage{pgfplots}
\usepgfplotslibrary{fillbetween}
\usetikzlibrary{bayesnet, quotes, decorations.pathreplacing,calligraphy}
\usepackage{csquotes}

\usepackage{caption}
\usepackage{subcaption}

\renewcommand\AB@authnote[1]{\rlap{\textsuperscript{\normalfont#1}}}

\newtheorem{lemma}{Lemma}
\newtheorem{theorem}{Theorem}
\newtheorem{proposition}{Proposition}

\newtheorem{corollary}{Corollary}

\theoremstyle{definition}

\newcommand{\R}{\mathbb{R}}

\DeclareMathOperator{\E}{\mathrm{E}}

\DeclareMathOperator{\Var}{\mathrm{Var}}

\title{The Moderating Effect of Instant Runoff Voting\thanks{The extended version of a AAAI '24 paper.}}
\author[1]{Kiran Tomlinson\thanks{kt@cs.cornell.edu}} %  arxiv
\author[2]{Johan Ugander}
\author[1]{Jon Kleinberg}
\affil[1]{Department of Computer Science, Cornell University}
\affil[2]{Department of Management Science and Engineering, Stanford University}
\date{}

\date{}
\begin{document}
\maketitle 

\begin{abstract}
Instant runoff voting (IRV) has recently gained popularity as an alternative to plurality voting for political elections, with advocates claiming a range of advantages, including that it produces more moderate winners than plurality and could thus help address polarization. However, there is little theoretical backing for this claim, with existing evidence focused on case studies and simulations. In this work, we prove that IRV has a moderating effect relative to plurality voting in a precise sense, developed in a 1-dimensional Euclidean model of voter preferences.  We develop a theory of {\it exclusion zones}, derived from properties of the voter distribution, which serve to show how moderate and extreme candidates interact during IRV vote tabulation. The theory allows us to prove that if voters are symmetrically distributed and not too concentrated at the extremes, IRV cannot elect an extreme candidate over a moderate. In contrast, we show plurality can and validate our results computationally. Our methods provide new frameworks for the analysis of voting systems, deriving exact winner distributions geometrically and establishing a connection between plurality voting and stick-breaking processes.

\end{abstract}

\section{Introduction}

Instant runoff voting (IRV) elections ask voters to rank candidates in order of preference and use a sequence of ``instant runoffs'' to determine a winner.\footnote{IRV is also called ranked choice voting in the United States. Other names for IRV include alternative vote, preferential voting, and the Hare method. Multi-winner IRV is also called single transferrable vote. Plurality is also called first-past-the-post.}
IRV selects a winner by repeatedly eliminating the candidate with the fewest first-place votes, redistributing those votes to the next-ranked candidate on each ballot, and removing the eliminated candidate from all ballots. The final remaining candidate is declared the winner (equivalently, one can terminate when a majority of the remaining ballots list the winner first).   
By comparison, in a plurality election the winner is simply the candidate with the most first-place votes. While plurality has historically been the predominant single-winner voting system, IRV is among the most popular alternatives; for instance, Australia and Ireland have used IRV  since the early 20th century. In the United States, IRV has recently been gaining traction to address issues with plurality voting~\cite{wang2021systems}, with three states (Maine, Alaska, and Nevada) voting to adopt IRV for federal elections in the last decade. IRV has also seen increasing adoption in local elections and/or primaries, for instance in San Francisco (since 2004), Minneapolis (since 2009), and New York City (since 2021). 
 
Proponents of IRV claim that it encourages moderation, compromise, and
civility, since candidates are incentivized to be ranked highly by as
many voters as possible, including by those who do not rank 
them first~\cite{dean2016,diamond2016}. Analyses of
campaign communication materials and voter surveys have supported the
theory that IRV increases campaign
civility~\cite{donovan2016campaign,john2017candidate,kropf2021using},
with extensive debate about whether this greater civility 
translates into winners who are also more moderate in their positions
\cite{fraenkel2006does,
fraenkel2006failure,horowitz2006strategy,horowitz2007have}.
Analyses of potential moderating effects of IRV have primarily been 
based on case studies~\cite{fraenkel2004neo,mitchell2014single,reilly2018centripetalism} 
and simulation~\cite{chamberlin1978toward,merrill1984comparison,mcgann2002party}, as well as 
empirical evidence for a moderating effect 
in a related voting system, two-round runoff~\cite{bordignon2016moderating}. In contrast, there has been almost no theoretical work on the subject; most social choice theory has focused on problems other than moderation, such as minimizing metric distortion and ensuring fairness or representation~\cite{halpern2023representation,aziz2017justified,boutilier2012optimal,brill2022individual,ebadian2022optimized,gkatzelis2020resolving,kahng2023voting}. Two interesting specific exceptions can be found in the works of \citet{grofman2004if} and \citet{dellis2017policy}. \citet{grofman2004if} show that for single-peaked preferences and
four or fewer candidates, IRV is at least as likely as plurality
to elect the median candidate. \citet{dellis2017policy} show that in a citizen-candidate model, if the voter distribution is asymmetric then two-party equilibria under plurality can be more extreme than under IRV.

There is clear value in mathematical analyses that identify more general moderating tendencies.
At present---beyond the noted exceptions---the arguments for IRV's moderating effects summarized above
have tended to point to institutional or behavioral properties of the
way candidates run their campaigns in IRV elections.
A natural question, therefore, is whether this picture is complete,
or whether there might be something in the definition of IRV itself
that leads to outcomes with more moderate winners.
Such questions are fundamental to the mathematical theory of voting
more generally, where we frequently seek explanations that are rooted
in the formal properties of the voting systems themselves, rather than 
simply the empirical regularities of how candidates and voters tend
to behave in these systems.
In the case of IRV, what would it mean to formalize a tendency toward
moderation in the underlying structure of the voting system?
To begin, we must first identify a natural set of definitions under which we can
isolate such a property.

\paragraph{Formalizing the moderating effect of IRV.}
In this paper, we propose such definitions and use them to articulate
a precise sense in which IRV produces moderate winners in a way
that plurality does not.
We work within a standard one-dimensional model of voters and candidates: 
the positions of voters and candidates correspond to points
drawn from distributions on the unit interval $[0,1]$ of the real line (representing left--right ideology), 
and voters form preferences over candidates
by ranking them in order of proximity. That is, voters favor candidates who are
closer to them on the line; this is often called the \textit{1-Euclidean} model, a common model in social choice theory~\cite{coombs1964theory,bogomolnaia2007euclidean,elkind2022restrictions}. We typically assume the voters and candidates are drawn from the same distribution $F$, but some of our results hold for fixed candidate positions.
In addition to its role as one of the classical mathematical models
of voter preferences, where it is sometimes called the Hotelling model~\cite{hotelling1929stability,downs1957economic}, 1-Euclidean preferences arise naturally from 
higher-dimensional opinions under simple models of opinion updating
\cite{demarzo2003persuasion}. There is wide-ranging empirical evidence suggesting that political
opinions in the United States are remarkably
one-dimensional~\cite{poole1984polarization,poole1991patterns,layman2006party,dellaposta2015liberals}:
from a voter's views on any one of a set of issues including
tax policy, immigration, 
climate change, gun control, and abortion, it is possible to 
predict the others with striking levels of confidence.

Let's consider a voting system applied to a set of $k$ candidates
and a continuum of voters in this setting: we draw $k$ candidates
independently from a given distribution $F$ on the unit interval $[0,1]$,
and each candidate gets a vote share corresponding to the fraction of
voters who are closest to them (see \Cref{fig:example-dsns} for examples).
The use of a one-dimensional model gives a natural interpretation to the
distinction between moderate and extreme candidates: 
a candidate is more extreme if they are closer to the endpoints of 
the unit interval $[0,1]$.
We take two approaches to defining a moderating effect in this model, one probabilistic (in the limit of large $k$) and one combinatorial (for all $k$).
We say that a voting system has a {\em probabilistic moderating effect} if 
for some interval $I = [a,b]$ with $0 < a \leq b < 1$,
the probability that the winning candidate comes from $I$ 
converges to 1 as the number of candidates $k$ goes to infinity (since we focus on symmetric voter distributions, we will typically have $I$ symmetric about $1/2$; i.e. $b = 1 - a$). We say that a voting system has a {\em combinatorial moderating effect} if for all $k$, the presence of a candidate in $I$ prevents any candidate outside of $I$ from winning; i.e., a moderate candidate (inside $I$) is guaranteed to win as long as at least one moderate runs. (Note that a combinatorial moderating effect implies a probabilistic one, as long as the candidate distribution $F$ places positive probability mass on $I$.) We call such an interval $I$ an \textit{exclusion zone} of the voting system, since the presence of a candidate inside this zone precludes outside candidates from winning.
In this way, a voting system with a moderating effect will tend
to suppress extreme candidates who lie
outside a middle portion  of the unit interval,
while a voting system that does not have at moderating effect will
allow arbitrarily extreme candidates to win with positive probability
even as the number of candidates becomes large.

Using this terminology, we can state our first main result succinctly:
under a uniform voter distribution, IRV has a moderating effect
and plurality does not---in both the combinatorial and probabilistic senses.
In particular, we prove a novel and striking fact about IRV:
when voters and candidates both come from the uniform distribution on $[0,1]$,
the probability that the winning candidate produced by IRV lies
outside the interval $[1/6,5/6]$ goes to 0 as the number of candidates
$k$ goes to infinity. 
In sharp contrast, the distribution of the plurality winner's position converges to uniform as the number of candidates goes to infinity,  allowing arbitrarily extreme candidates to win.
As part of our analysis,
we provide a method for deriving the distribution of plurality and IRV winner positions for finite $k$ and perform this derivation for $k=3$ candidates. Surprisingly, our analysis of plurality---the simpler voting system---requires much more sophisticated machinery: we establish a connection between plurality voting and a classic model in discrete probability known as the stick-breaking process and develop new asymptotic stick-breaking results for use in our analysis. 

\begin{figure}
\centering
  \begin{tikzpicture}
\begin{axis}[height=5cm, width=6cm, xmin=0, xmax=1,ymin=-0.1, ymax=2, samples=500, axis y line=none,       axis x line=middle, yticklabels={..}, axis line style=thick]
  \addplot[black, thick, name path=A][domain=0.001:0.999] (x, {(x^(1))*(1-x)^(1)*6});
  \addplot[draw=none,name path=B] {0};     
  \addplot[fill={rgb,255:red,255; green,200; blue,200},fill opacity=1] fill between[of=A and B,soft clip={domain=0:0.25}];
  \addplot[fill={rgb,255:red,255; green,150; blue,150},fill opacity=1] fill between[of=A and B,soft clip={domain=0.25:0.35}];
  \addplot[fill={rgb,255:red,255; green,100; blue,100},fill opacity=1] fill between[of=A and B,soft clip={domain=0.35:0.625}];
  \addplot[fill={rgb,255:red,255; green,50; blue,50},fill opacity=1] fill between[of=A and B,soft clip={domain=0.625:1}];
  \node[label={90:{A}},circle,fill={rgb,255:red,255; green,200; blue,200},inner sep=2pt, draw=black, line width=0.2mm] at (axis cs:0.2, 0) {};
  \node[label={90:{B}},circle,fill={rgb,255:red,255; green,150; blue,150},inner sep=2pt, draw=black, line width=0.2mm] at (axis cs:0.3, 0) {};
  \node[label={90:{C}},circle,fill={rgb,255:red,220; green,100; blue,100},inner sep=2pt, draw=black, line width=0.2mm] at (axis cs:0.4, 0) {};
  \node[label={90:{D}},circle,fill={rgb,255:red,255; green,50; blue,50},inner sep=2pt, draw=black, line width=0.2mm] at (axis cs:0.85, 0) {};
%    \node[] at (axis cs:0.5, 1.7) {Beta(2,2)};
      \node[] at (axis cs:0.16, .46) {\scriptsize 0.16};
    \node[] at (axis cs:.3, .46) {\scriptsize 0.13};
    \node[] at (axis cs:.5, .46) {\scriptsize 0.40};
    \node[] at (axis cs:.8, .46) {\scriptsize 0.32};
\end{axis};
\end{tikzpicture}
\quad
  \begin{tikzpicture}
\begin{axis}[height=5cm, width=6cm, xmin=0, xmax=1,ymin=-0.1, ymax=2, samples=200, axis y line=none,       axis x line=middle, yticklabels={..}, axis line style=thick]
  \addplot[black, thick, name path=A][domain=0.001:0.999] (x, {1});
  \addplot[draw=none,name path=B] {0};     
  \addplot[fill={rgb,255:red,255; green,200; blue,200},fill opacity=1] fill between[of=A and B,soft clip={domain=0:0.25}];
  \addplot[fill={rgb,255:red,255; green,150; blue,150},fill opacity=1] fill between[of=A and B,soft clip={domain=0.25:0.35}];
  \addplot[fill={rgb,255:red,255; green,100; blue,100},fill opacity=1] fill between[of=A and B,soft clip={domain=0.35:0.625}];
  \addplot[fill={rgb,255:red,255; green,50; blue,50},fill opacity=1] fill between[of=A and B,soft clip={domain=0.625:1}];
  \node[label={90:{A}},circle,fill={rgb,255:red,255; green,200; blue,200},inner sep=2pt, draw=black, line width=0.2mm] at (axis cs:0.2, 0) {};
  \node[label={90:{B}},circle,fill={rgb,255:red,255; green,150; blue,150},inner sep=2pt, draw=black, line width=0.2mm] at (axis cs:0.3, 0) {};
  \node[label={90:{C}},circle,fill={rgb,255:red,220; green,100; blue,100},inner sep=2pt, draw=black, line width=0.2mm] at (axis cs:0.4, 0) {};
  \node[label={90:{D}},circle,fill={rgb,255:red,255; green,50; blue,50},inner sep=2pt, draw=black, line width=0.2mm] at (axis cs:0.85, 0) {};
%      \node[] at (axis cs:0.5, 1.7) {Beta(1,1)};
  \node[] at (axis cs:0.125, .46) {\scriptsize 0.25};
    \node[] at (axis cs:.3, .46) {\scriptsize 0.10};
    \node[] at (axis cs:.5, .46) {\scriptsize 0.28};
    \node[] at (axis cs:.8, .46) {\scriptsize 0.38};
\end{axis};
\end{tikzpicture}
\quad
  \begin{tikzpicture}
\begin{axis}[height=5cm, width=6cm, xmin=0, xmax=1,ymin=-0.1, ymax=2, samples=500, axis y line=none,       axis x line=middle, yticklabels={..}, axis line style=thick]
      
      \fill [color={rgb,255:red,255; green,200; blue,200}] (axis cs:0,0) rectangle (axis cs:0.027,2);
      
      \fill [color={rgb,255:red,255; green,50; blue,50}] (axis cs:0.974,0) rectangle (axis cs:1,2);
      
  \addplot[black, thick, name path=A][domain=0.001:0.999] (x, {min((x^(-0.5))*(1-x)^(-0.5)/3.14159, 3)});
  \addplot[draw=none,name path=B] {0};     
  \addplot[fill={rgb,255:red,255; green,200; blue,200},fill opacity=1] fill between[of=A and B,soft clip={domain=0:0.25}];
  \addplot[fill={rgb,255:red,255; green,150; blue,150},fill opacity=1] fill between[of=A and B,soft clip={domain=0.25:0.35}];
  \addplot[fill={rgb,255:red,255; green,100; blue,100},fill opacity=1] fill between[of=A and B,soft clip={domain=0.35:0.625}];
  \addplot[fill={rgb,255:red,255; green,50; blue,50},fill opacity=1] fill between[of=A and B,soft clip={domain=0.625:1}];
  \node[label={90:{A}},circle,fill={rgb,255:red,255; green,200; blue,200},inner sep=2pt, draw=black, line width=0.2mm] at (axis cs:0.2, 0) {};
  \node[label={90:{B}},circle,fill={rgb,255:red,255; green,150; blue,150},inner sep=2pt, draw=black, line width=0.2mm] at (axis cs:0.3, 0) {};
  \node[label={90:{C}},circle,fill={rgb,255:red,220; green,100; blue,100},inner sep=2pt, draw=black, line width=0.2mm] at (axis cs:0.4, 0) {};
  \node[label={90:{D}},circle,fill={rgb,255:red,255; green,50; blue,50},inner sep=2pt, draw=black, line width=0.2mm] at (axis cs:0.85, 0) {};
%      \node[] at (axis cs:0.5, 1.7) {Beta(0.5,0.5)};
  \node[] at (axis cs:0.125, .46) {\scriptsize 0.33};
    \node[] at (axis cs:.3, .46) {\scriptsize 0.07};
    \node[] at (axis cs:.5, .46) {\scriptsize 0.18};
    \node[] at (axis cs:.8, .46) {\scriptsize 0.42};
\end{axis};
\end{tikzpicture}
\caption{Three example voter distributions in one dimension (all Betas). Candidates A, B, C, D are positioned at 0.2, 0.3, 0.4, and 0.85. The black line shows the density function of the voter distribution. Regions are colored according to the most preferred candidate of voters in that region and annotated with the approximate vote share of that candidate. As an example, the preference ordering of a voter at 0.5 is C, B, A, D (regardless of the voter distribution). Similarly, a voter at 0.1 has preference ordering A, B, C, D. In the moderate voters example (left), C is both the plurality and IRV winner.  In the uniform voters example (center), D is the plurality winner and C is the IRV winner. In the polarized voters example (right), D is the plurality winner and A is the IRV winner.}\label{fig:example-dsns}
\end{figure}
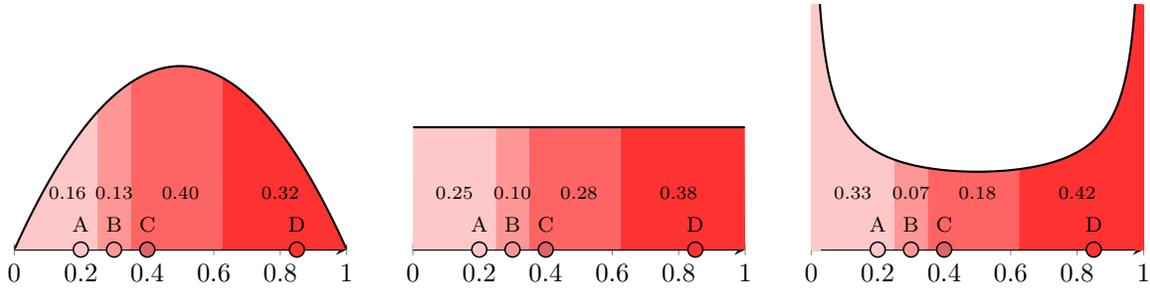

Our probabilistic result for IRV follows from a companion fact that is 
combinatorial in nature and comparably succinct: 
given any finite set of candidates in $[0,1]$,
and voters from the uniform distribution, if any of the candidates
belong to the interval $[1/6,5/6]$, then the IRV winner must come
from $[1/6,5/6]$; that is, $[1/6,5/6]$ is an exclusion zone for IRV in the uniform case.
Moreover, $[1/6,5/6]$ is the smallest interval for which this statement
is true.
Again, the analogue for plurality voting with any proper sub-interval of 
the unit interval is false: we show that plurality has no exclusion zones. 

This first main result therefore gives a precise sense in which the
structure of the IRV voting system favors moderate candidates:
whenever moderate candidates (in the middle two-thirds of the unit interval)
are present as options, IRV will push out more extreme candidates.
We then address the more challenging case of non-uniform voter distributions, where we prove that IRV continues to have a moderating effect (in the sense
of our formal definitions) even for voter distributions that push probability
mass out toward the extremes of the unit interval, up to a specific
threshold beyond which the effects cease to hold.
Thus, IRV is even able to offset a level of polarization built into
the underlying distribution of voters and candidates, although it can only
do so up until a certain level of polarization is reached. In contrast, we establish 
that plurality never has a combinatorial moderating effect for any non-pathological voter distribution.

As a final point, it is worth emphasizing what is and is not a focus of our work here.
We examine IRV and plurality because of their widespread use in real-world elections and the fierce debate surrounding the adoption of IRV over plurality. 
We are not trying to characterize all possible voting systems that give rise to moderation (although we can show that many voting systems not in widespread use have a moderating effect, including the Coombs rule and any Condorcet method; for these systems, any symmetric interval  around 0.5 is an exclusion zone).
Our interest, instead, is in the following contribution to the plurality--IRV debate: there is a precise mathematical sense in which IRV has a moderating effect and plurality does not.
Second, we do not analyze strategic choices by candidates about
where to position themselves on the unit interval
\cite{hotelling1929stability,downs1957economic,osborne1995spatial},
but instead derive properties of voting systems that hold for
fixed candidate positions, or candidate positions drawn from a distribution.
This approach produces results that are robust
against the question of whether candidates are actually able
to make optimal strategic positioning decisions in practice~\cite{bendor2011behavioral};
it also allows us to better understand how the voting
systems themselves behave---providing a foundation for future
strategic work.

\section{Uniform voters}\label{sec:uniform}

The previous section describes our complete model, but it is
useful to review it here in the context of some more specific notation.
We assume voters and candidates are both drawn from a distribution $F$
on the unit interval $[0,1]$,
representing their ideological position on a left--right 
spectrum.\footnote{We will generally focus on distributions $F$
that are symmetric around $1/2$ and represented by a density function $f$.}
Voters prefer candidates closer to them (i.e., they have 1-Euclidean preferences).
There are $k$ candidates  drawn independently from $F$;
suppose that these draws produce candidate positions 
$x_1 < x_2 < \dots < x_k$ in order. Some of our results apply regardless of the candidate distribution, relying only on the voter distribution; we will make a note of such cases.

Since we want to model the case of a large population of voters,
we do not explicitly sample the voters from $F$, but instead
think of a continuum of voters who correspond to the distribution
$F$ itself: that is, under the plurality voting rule, the
fraction of voters who vote for candidate $x_i$ is the probability
mass of all voters who are closer to $x_i$ than to any other candidate
(or, equivalently, it is the probability that a voter randomly
chosen according to $F$ would be closer to $x_i$ than any other candidate). In this section, we focus on the case where $F$ is uniform.\footnote{To provide another perspective on the uniform voter assumption, consider the following preference assumption that also produces uniform 1-Euclidean preferences: voters are arbitrarily distributed, but rank candidates according to how many voters are between them and each candidate. That is, voters have 1-Euclidean preferences \textit{in the voter quantile space} and are always uniformly distributed over this space by definition. All of our uniform voter results hold in that setting as well, although stated in terms of voter quantiles rather than absolute positions.}
We use $v(x_i)$ to denote the vote share for candidate $x_i$.
Under IRV, the candidate $i$ with the smallest
$v(x_i)$ is eliminated and vote shares are recomputed without
candidate $i$. This repeats until only one candidate remains, who is
declared the winner (equivalently, elimination can terminate when a
candidate achieves majority). In practice, voters submit a ranking
over the candidates and their votes are ``instantly'' redistributed
after each elimination.

\paragraph{IRV's moderating effect: A first result.}
With uniform 1-Euclidean voters, we now show that IRV cannot elect extreme
candidates over moderates---regardless of the distribution of candidates. 
That is, IRV exhibits an exclusion zone in the middle of the unit interval, 
where the presence of moderate candidates inside the zone precludes 
outside extreme candidates from winning. The idea behind the proof is that as
moderates get eliminated, the middle part of the interval becomes
sparser, granting a higher vote share to any remaining moderates. Consider the moment when only one candidate $x$ remains in the interval $[1/6, 5/6]$ (see \Cref{fig:1/6-proof}); extreme
candidates near 0 and 1 are then too far away to ``squeeze out'' $x$. With uniform voters, the tipping point for
squeezing out moderates occurs when
extreme candidates are at 1/6 and 5/6. In the next section, we
present generalizations of this result for non-uniform voter
distributions.

\begin{theorem}\label{thm:1/6}
(Combinatorial moderation for uniform IRV.) Under IRV with uniform voters over $[0, 1]$ and $k\ge 3$ candidates, if there is a candidate in $[1/6, 5/6]$, then the IRV winner is in $[1/6, 5/6]$. No smaller interval $[c, 1-c]$, $c>1/6$, has this property. If there are no candidates in $[1/6, 5/6]$, then the IRV winner is the one closest to $1/2$.
\end{theorem}

\begin{proof}
Suppose first there is only one candidate $x \in [1/6, 5/6]$ and all other candidates are $< 1/6$ or $> 5/6$. Suppose without loss of generality that $x \le 1/2$. The smallest vote share $x$ could have occurs when there are candidates at $1/6-\epsilon$ and $5/6+\epsilon$. In this case, $x$ gets vote share $(x - 1/6 + \epsilon) / 2 + (5/6 + \epsilon - x)/2 = 1/3+ \epsilon$. Meanwhile, the highest vote share any candidate $< 1/6$ could have (when $x$ is at $1/2$) is $1 / 6 - \epsilon + (1/2 - 1/6 + \epsilon) / 2 = 1/3 - \epsilon / 2$. Thus, every candidate $< 1/6$ will be eliminated before $x$. At this point, $x$ will win, since it is closer to $1/2$ than any of remaining candidates $> 5/6$ and therefore has a majority. 
  
If there is more than one candidate in $[1/6, 5/6]$ to begin with, then as candidates are eliminated, at some point there will only be one candidate $x$ remaining in $[1/6, 5/6]$. Either $x$ will be the ultimate winner, or there will still be candidates $<1/6$ or $> 5/6$, in which case $x$ will win as argued above.
  
Notice that the above argument still holds if we replace $1/6$ and $5/6$ with $c$ and $1-c$ for any $0 < c \le 1/6$: it only reduces the vote share going towards candidates in $[0, c)$. Thus, if there is some candidate in $[c, 1-c]$, then the IRV winner is in $[c, 1-c]$. So, if there is no candidate in $[1/6/, 5/6]$, then let $c$ be the distance between the most moderate candidate and its closest edge. This candidate must be the IRV winner, since it is the only candidate in $[c, 1-c]$ (and $c < 1/6$). In this case, the IRV winner is the most moderate candidate as claimed. 
  
Finally, we show that no smaller interval satisfies the theorem.
To do so, we describe a construction that is parametrized to handle
any number of candidates $k \geq 3$.
In the construction, 
there is one candidate at $1/2$, two
candidates at $c$ and $1-c$ for $c > 1/6$, and any remaining candidates in $(1-\epsilon, 1]$ for $\epsilon < (1/2 - c)/2$. First, all candidates right of $1-\epsilon$ will be eliminated---all of these candidates have a smaller vote share than the candidate at $c$. At this point, the
candidate at $1/2$ has vote share less than $2(1/2 - c)/2 = 1/2-c$. Since
$c > 1/6$, this is less than $1/2-1/6 = 1/3$. Meanwhile, the
candidates at $c$ and $1-c$ have vote shares higher than $c + (1/2-c)/2 = 1/4+c/2$. Since $c > 1/6$, this is greater than $1/4+1/12 = 1/3$. Thus
the middle candidate is eliminated. The winner is is thus outside of $[c + \delta, 1-c-\delta]$ for all $\delta\in (0, 1/2-c)$, despite there being a candidate in this interval. Since this construction applies for any $c > 1/6$, no interval smaller than $[1/6, 5/6]$ satisfies the theorem.
\end{proof}
\begin{figure}
\centering
\pgfmathsetseed{1}
  \begin{tikzpicture}[x=4.5cm]
    \draw[line width=0.3mm] (0, 0) -- (1, 0);
    \draw[line width=0.3mm] (0, -0.1) -- (0, 0.1);
    \draw[line width=0.3mm] (1, -0.1) -- (1, 0.1);
    \draw[line width=0.3mm] (1/6, -0.1) -- (1/6, 0.1);
    \draw[line width=0.3mm] (5/6, -0.1) -- (5/6, 0.1);
    
    \node (0) at (0, -0.4) {$0$};
    \node (1/6) at (1/6, -0.4) {$1/6$};
    \node (5/6) at (5/6, -0.4) {$5/6$};
    \node (1) at (1, -0.4) {$1$};
    
\foreach \x in {0, ..., 10}
        {\pgfmathsetmacro{\xcoor}{(\x+ rnd-0.5) / 65 }
        \node[draw, circle, inner sep=1pt, fill] (\x) at (\xcoor, 0) {};}
\foreach \x in {0, ..., 10}
        {\pgfmathsetmacro{\xcoor}{5/6 + (\x+ rnd-0.5) / 61 }
        \node[draw, circle, inner sep=1pt, fill] (\x) at (\xcoor, 0) {};}
 \foreach \x in {0, ..., 40}
        {\pgfmathsetmacro{\xcoor}{1/6 + (\x+ rnd-0.5) / 60 }
        \node[draw, circle, inner sep=1pt, fill] (\x) at (\xcoor, 0) {};}       
        
          \node[draw, circle, inner sep=1.5pt, fill, label=$x$, color={rgb,255:red,220; green,0; blue,0}] (x) at (0.4, 0) {};
  \end{tikzpicture}
\raisebox{6mm}{$\rightarrow$}
\pgfmathsetseed{1}
\raisebox{-1.8mm}{
  \begin{tikzpicture}[x=4.5cm]
    \draw[line width=0.3mm] (0, 0) -- (1, 0);
    \draw[line width=0.3mm] (0, -0.1) -- (0, 0.1);
    \draw[line width=0.3mm] (1, -0.1) -- (1, 0.1);
    \draw[line width=0.3mm] (1/6, -0.1) -- (1/6, 0.1);
    \draw[line width=0.3mm] (5/6, -0.1) -- (5/6, 0.1);
    
    \node (0) at (0, -0.4) {$0$};
    \node (1/6) at (1/6, -0.4) {$1/6$};
    \node (5/6) at (5/6, -0.4) {$5/6$};
    \node (1) at (1, -0.4) {$1$};
    
    \node[draw, circle, inner sep=1.5pt, fill, label=$x$, color={rgb,255:red,255; green,50; blue,50}] (x) at (0.4, 0) {};
    \draw[line width=0.8mm, color={rgb,255:red,255; green,50; blue,50}] (17/60, 0) -- (37/60, 0);
    \draw [decorate,decoration={brace,amplitude=5pt,mirror,raise=4pt},yshift=0pt, line width=0.2mm]
(17/60, 0) -- (37/60, 0) node [black,midway, yshift=-6mm] {\footnotesize $v(x) > 1/3$};
\foreach \x in {0, ..., 4}
        {\pgfmathsetmacro{\xcoor}{(\x+ rnd-0.5) / 25 }
        \node[draw, circle, inner sep=1pt, fill] (\x) at (\xcoor, 0) {};}
\foreach \x in {0, ..., 4}
        {\pgfmathsetmacro{\xcoor}{5/6 + (\x+ rnd) / 30}
        \node[draw, circle, inner sep=1pt, fill] (\x) at (\xcoor, 0) {};}
  \end{tikzpicture}}
\raisebox{6mm}{$\rightarrow$}
\pgfmathsetseed{1}
  \begin{tikzpicture}[x=4.5cm]
    \draw[line width=0.3mm] (0, 0) -- (1, 0);
    \draw[line width=0.3mm] (0, -0.1) -- (0, 0.1);
    \draw[line width=0.3mm] (1, -0.1) -- (1, 0.1);
    \draw[line width=0.3mm] (1/6, -0.1) -- (1/6, 0.1);
    \draw[line width=0.3mm] (5/6, -0.1) -- (5/6, 0.1);
    
    \node (0) at (0, -0.4) {$0$};
    \node (1/6) at (1/6, -0.4) {$1/6$};
    \node (5/6) at (5/6, -0.4) {$5/6$};
    \node (1) at (1, -0.4) {$1$};
    
    \node[draw, circle, inner sep=1.5pt, fill, label=$x$, color={rgb,255:red,255; green,50; blue,50}] (x) at (0.4, 0) {};
    \draw[line width=0.8mm, color={rgb,255:red,255; green,50; blue,50}] (0, 0) -- (37/60, 0);
\foreach \x in {0}
        {\pgfmathsetmacro{\xcoor}{5/6 + (\x+ rnd) / 30}
        \node[draw, circle, inner sep=1pt, fill] (\x) at (\xcoor, 0) {};}
  \end{tikzpicture}
  \caption{Visual depiction of the proof of \Cref{thm:1/6}. IRV eliminates candidates until a final candidate $x$ remains in the exclusion zone $[1/6, 5/6]$. At this point, $x$ gets more than $1/3$ of the vote share and cannot be eliminated next (regardless of where they are in $[1/6, 5/6]$). Candidates outside of $[1/6, 5/6]$ are thus eliminated until $x$ wins.}\label{fig:1/6-proof}
\end{figure}
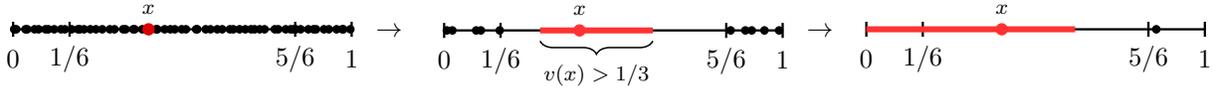

In the language of our analysis, $[1/6,5/6]$ is then the smallest possible exclusion zone of IRV under a uniform voter distribution. See \Cref{fig:1/6-proof} for a visual depiction of the argument.  A corollary of \Cref{thm:1/6} is that if candidates are distributed uniformly at random (for instance, if voters independently and identically decide whether to run for office), then IRV elects extreme candidates with probability going to 0 as the number of candidates grows, since the probability of having no moderate candidates in $[1/6, 5/6]$ is $(1/3)^k$. In the language defined earlier, IRV thus has a probabilistic moderating effect with uniform voters and candidates.

\begin{corollary}\label{cor:probabilistic-moderation}
(Probabilistic moderation for uniform IRV.) Let $R_k$ be the position of the IRV winner with $k$ candidates distributed uniformly at random and uniform voters. 
  \begin{equation}
    \lim_{k \rightarrow \infty} \Pr(R_k \notin [1/6, 5/6]) =0.
      \end{equation}
      
\end{corollary}

In contrast to IRV, where the presence of candidates with moderate positions (namely, inside $[1/6, 5/6]$) precludes extreme candidates from winning, we now show that no such fact is true for plurality (excluding the extreme points 0 and 1): for any interval $I \subseteq (0, 1)$, there is some configuration of candidates such that the winner is outside of $I$ despite having candidates in $I$. In other words, plurality voting does not have a combinatorial moderating effect with uniform voters.\footnote{An anonymous reviewer suggested an elegant construction proving this fact for symmetric intervals $I = [c, 1-c]$, which provides counterexamples for every $k \ge 3$: place candidates at $c-\epsilon, 1-c-\epsilon,$ and any others at $1-c + \epsilon$ (for $\epsilon < c / 2$). The candidate at $c-\epsilon$ wins, despite having a candidate in $I$.} 
Later, we generalize this result to non-uniform voter distributions. The idea behind the proof is relatively straightforward: given a set of candidates, keep adding candidates to reduce the vote share of everyone except the desired winner. 

\begin{theorem}\label{thm:no-exclusion-uniform}
(No combinatorial moderation for uniform plurality.) Suppose voters are uniformly distributed over $[0, 1]$. Given any set of $\kappa \ge 1$ distinct candidate positions $x_1, \dots, x_\kappa$ with $x_1 \notin \{0, 1\}$, there exists a configuration of $k\ge \kappa$ candidates (including $x_1, \dots, x_\kappa$) such that the candidate at $x_1$ wins under plurality.
\end{theorem}
\begin{proof}
We show how to add candidates to the initial set $x_1, \dots, x_\kappa$ so that $x_1$ becomes the plurality winner (as long as $x_1 \notin \{0, 1\}$). First, add candidates at $x_0=0$ and $x_{\kappa+1} = 1$ to guarantee that $x_1$ is between two candidates. Let $x_\ell$ be the candidate to the left of $x_1$ and let $x_r$ be the candidate to the right of $x_1$. Let $v_\ell = (x_1 - x_\ell)/ 2$ be the vote share $x_1$ gets on its left and let $v_r = (x_r - x_1)/2$ be the vote share $x_1$ gets on its right. Add new candidates spaced by $\frac{1}{2} \min\{v_\ell, v_r\}$ in the intervals $[0, x_\ell]$ and $[x_r, 1]$. This causes every candidate in the intervals $[0, x_\ell)$ and $(x_r, 1]$ to have vote share strictly less than $\frac{1}{2}\min\{v_\ell, v_r\}$ (whether they are part of the original $\kappa$ or new). Additionally, $x_\ell$ and $x_r$ have vote share at most $\frac{1}{2} \min\{v_\ell, v_r\} + \max\{v_\ell, v_r\}$. Meanwhile, $x_1$ has vote share $v_\ell + v_r$, so $x_1$ is the plurality winner in the new configuration.\end{proof}

In addition, we prove that the asymptotic distribution of the plurality winner's position is uniform over the unit interval when voters and candidates are positioned uniformly at random. In other words, plurality does not have a probabilistic moderating effect: it does not preclude extreme candidates from winning when there are many moderate candidates to choose from. The proof is more involved, so we relegate it to \Cref{app:proofs}. Note that this result implies plurality also has no combinatorial moderation, but \Cref{thm:no-exclusion-uniform} is much easier to prove. 

\begin{theorem}\label{thm:plurality-asymptotic-uniform}
(No probabilistic moderation for uniform plurality.)  Let $P_k$ be the position of the plurality winner with $k$ candidates distributed uniformly at random and uniform voters. As $k \rightarrow \infty$, $P_k$ converges in distribution to Uniform$(0, 1)$; that is, $\lim_{k \rightarrow \infty} \Pr(P_k \le x) = x$ for all $x \in [0, 1]$.
\end{theorem}

The proof uses a coupling argument between plurality on the unit interval and plurality on a circle. By rotational symmetry, the plurality winner on a circle is uniformly distributed. We show that as $k$ grows, cutting the circle to transform it into the interval does not change the winner with probability approaching 1, since cutting the circle only affects vote shares of the boundary candidates. 

Thus, a key step is deriving the asymptotic distribution of the
winning plurality vote share. This vote share distribution may be useful for other asymptotic analyses of plurality voting, so we describe it here. The winning plurality vote share is closely related to a category
of probabilistic problems known as 
{\em stick-breaking problems}, which focus on the properties of a stick of
length $1$ broken into $n$ pieces uniformly at
random~\cite{holst1980lengths}. Setting $n = k+1$, these stick pieces
can be viewed as the gaps between candidates (equivalently, candidates
are the breakpoints of the stick). A classic result in stick-breaking
is that the biggest piece will have size $B_n$ almost exactly $\log n
/ n$ as $n$ grows large~\cite{darling1953class,holst1980lengths} and
that $nB_n - \log n$ converges to a Gumbel$(1, 0)$ distribution as $n
\rightarrow \infty$. The plurality vote setting is different,
since candidates get vote shares from half of the gap to their left plus
half of the gap to their right (except the left- and rightmost
candidates). We show that as the
number of candidates grows large, the winning vote share $V_k$ with $k
= n-1$ candidates is almost exactly $(\log n + \log \log n)/2n$ and
that $n V_k - (\log n + \log \log n) / 2$ also converges to Gumbel$(1,
0)$ as $k \rightarrow \infty$. Intuitively, the largest pair of
adjacent gaps have size $\log n /n$ and $\log \log n /n$, and the
candidate between these gaps gets vote shares from half of each gap (more
correctly, the total size of the gaps is
$(\log n + \log \log n)/n$). This is formalized in the following lemma used to prove \Cref{thm:plurality-asymptotic-uniform}.

\begin{lemma}\label{lemma:gumbel}
Let $V_k$ be the winning plurality vote share with $k$ candidates distributed uniformly at random over $[0, 1]$ and uniform voters. Setting $n = k+1$, 
\begin{equation}
  \lim_{k\rightarrow \infty} \Pr\left(V_k \le \frac{\log n + \log \log n + x}{2n} \right) = e ^{-e^{-x}}.
\end{equation}
\end{lemma}

\subsection{Plurality and IRV winner distributions}
\begin{figure}[t]
\centering
  \includegraphics[width=\textwidth]{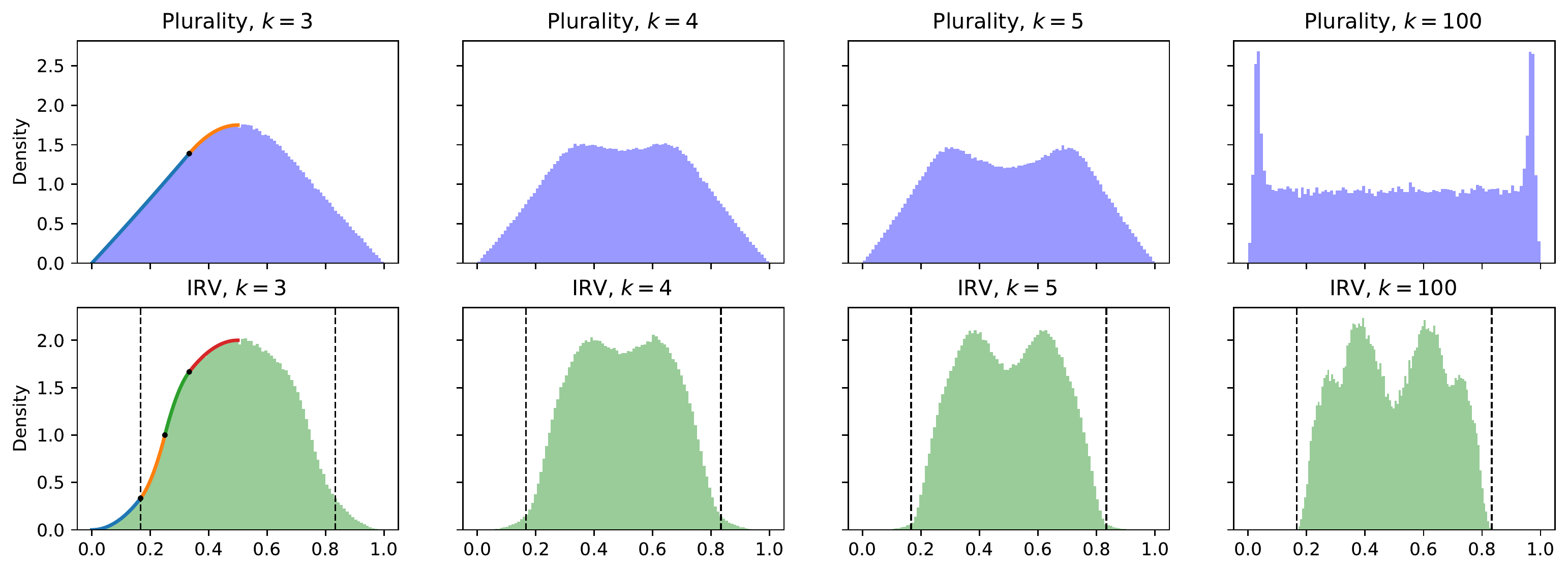}
  \caption{The distributions of the winning position with $k=3, 4, 5,$ and $100$ candidates and continuous 1-Euclidean voters (both uniformly distributed) under plurality and IRV. The histograms are from 1 million simulation trials for $k=3, 4, 5$ and 100000 trials for $k=100$, while the curves plotted for $k=3$ (shown up to 1/2) are the exact density functions given in \Cref{thm:plurality-dsn,thm:irv-dsn}, with pieces separated by color. Note that the IRV winner is only at a position $<1/6$ or $>5/6$ when no candidates fall in $[1/6, 5/6]$ by \Cref{thm:1/6}; the dashed vertical lines outline this exclusion zone. The probabilistic moderating effect for IRV is already strong at with only $k=5$ candidates.}\label{fig:k-3-distributions}
\end{figure}

Given these results about the asymptotic distributions of the plurality and IRV winner positions $P_k$ and $R_k$, asymptotic in the number of candidates $k$, a natural follow-on question is whether we can say anything about these distributions for fixed values of $k$.   

For a fixed value of $k$, the distributions of the plurality and IRV
winner positions $P_k$ and $R_k$ with uniform voters and candidates
have density functions $f_{P_k}$ and $f_{R_k}$ that are piecewise
polynomial of order $k-1$. To see this, consider a point in
the $k$-dimensional
unit hypercube, where dimension $i$ of this point
represents the position of
candidate $i$. For every left-right order of candidates $\pi \in S_k$
(where $\pi(i)$ is the index of candidate $i$ in left-right order and
$S_k$ is the symmetric group on $k$ elements), we can express the
region in $\R^k$ where candidate $i$ wins given order $\pi$ using the
following collection of linear inequalities:

\begin{align*}
  &0\le x_{\pi^{-1}(1)} < x_{\pi^{-1}(2)} < \dots < x_{\pi^{-1}(k)}\le 1,\\
  &v(x_i) = \frac{x_{r(i)}-x_{\ell(i)}}{2} > \frac{x_{r(j)}-x_{\ell(j)}}{2} = v(x_j), \tag{for all $j \ne i$}
\end{align*}
where $\ell(i) = \pi^{-1}(\pi(i) - 1)$ is the candidate to $i$'s left and $r(i) = \pi^{-1}(\pi(i) + 1)$ is the candidate to $i$'s right. The inequalities in the first line ensure the left--right candidate order matches $\pi$, while the inequalities on the second line ensure $x_i$ has a larger vote share than any other candidate (i.e., $x_i$ is the plurality winner). The region defined by these linear inequalities is therefore a convex polytope, as it is the intersection of a finite number of half spaces. 

To find the probability that a candidate $i$ at a particular point $x$ wins under plurality, we can find the sum of the cross-sectional areas of these polytopes at $x_i=x$ (with one polytope for each of the $k!$ candidate orderings), integrating over the positions of the other $k-1$ candidates. This procedure produces a piecewise polynomial in $x$ of order $k-1$, where pieces are split at the vertices of the polytopes. To convert the win probability of candidate $i$ at position $x$ into the winner position density at $x$, we scale by $k$ to account for the symmetry in choosing $i$.

 We can use this approach to derive $f_{P_3}$, the winner distribution for plurality with $3$ candidates (see \Cref{fig:k-3-distributions} for a visualization). As the derivation is tedious, we present it in \Cref{app:k3_derivations}. Additionally, \Cref{app:figures} includes a visualization of the winning position polyhedra for $k=3$ whose cross-sectional areas produce $f_{P_3}$. 
 
\begin{proposition}\label{thm:plurality-dsn} 
  \begin{equation}
    f_{P_3}(x) = \begin{cases}
   x^2/2 + 4 x, & x \in [0, 1/3]\\
  - 13 x^2  + 13 x -3/2 , & x \in [1/3, 1/2]\\
  f_{P_3}(1-x), & x \in (1/2, 1].
 \end{cases}
  \end{equation}
\end{proposition}

 For analyzing the plurality winner distribution in this way with larger $k$ (even after accounting for relabeling symmetry), we would need to integrate over $k$ $k$-polytopes, each of which has a number of faces growing linearly with $k$ (one face per inequality requiring that $x_i$ beats each other $x_j$). Unfortunately, the number of vertices per polytope in this procedure could grow exponentially with $k$, potentially requiring exponentially many integrals.

The same strategy can also be used for IRV, except we no longer have only one polytope per permutation of candidates---instead, we have one polytope per combination of left-right candidate order and candidate elimination order. If we fix both, the region where candidate $i$ wins under IRV can once again be defined by a collection of linear inequalities. We used this approach to derive the IRV winner distribution with 3 candidates, $f_{R_3}$. Again, see \Cref{fig:k-3-distributions} for a visualization of $f_{R_3}$, \Cref{app:k3_derivations} for the derivation, and \Cref{app:figures} for a visualization of the IRV polyhedra.

\begin{proposition}\label{thm:irv-dsn}
  
  \begin{equation}
    f_{R_3}(x) = \begin{cases}
  12x^2, & x \in [0, 1/6]\\
  48 x^2  - 12 x + 1, & x \in [1/6, 1/4]\\
   - 48 x^2+ 36 x-5 , & x \in [1/4, 1/3]\\
   - 12 x^2+ 12 x-1 , & x \in [1/3, 1/2]\\
  f_{R_3}(1-x), & x\in (1/2, 1].
   \end{cases}  \end{equation}

\end{proposition}

For IRV with general $k$, this analysis requires integrating over $k!$ $k$-polytopes, each of which has $O(k^2)$ faces: given an elimination order, we need an inequality specifying that the candidate eliminated $i^{\text{th}}$ has a smaller vote share than each of the candidates eliminated later. Each such inequality defining a half-space can add a face to the polytope. 

Note that in \Cref{thm:irv-dsn}, the integral of the density $f_{R_3}(x)$ on $[0,1/6]$ is exactly equal to half the probability that the $k-1$ losing candidates did not appear inside $[x,1-x]$ (scaled by $k$ to account for relabeling symmetry), since we know by \Cref{thm:1/6} that a candidate can only win outside $[1/6, 5/6]$ if they are the most moderate candidate. For general $k>3$ we can easily derive the density on $[0,1/6]$ and $[5/6,1]$ using the generalization of this argument: $f_{R_k}(x) =k(2x)^{k-1}$ on $[0,1/6]$ (with the right tail being mirrored). Note that the integral of $f_{R_k}(x)=k(2x)^{k-1}$ over $[0, 1/6]$ goes to 0 as $k\rightarrow \infty$, a limit that furnishes an independent way of establishing a probabilistic moderating effect for IRV.

Having the exact winner position distributions $f_{P_3}$ and $f_{R_3}$ allows us to answer additional questions---for instance, how much more moderate do IRV winners tend to be for $k=3$ with uniform voters and candidates? Using the density functions above, we can analytically compute the variances of the plurality and IRV winner distributions, $\Var(P_3) = 23/540$ and $\Var(R_3) = 25/864$. For $k=3$, the variance of the plurality winner's position with uniform voters is thus exactly $184/125 = 1.472$ times higher than the variance of the IRV winner's position.

Connecting our results to related work, while the distribution of the 
winner's position is challenging to derive, 
the expected plurality vote share at each point 
is more tractable. This distribution was discovered in another context: a guessing game where the goal is to be closest to
an unknown target distributed uniformly at random, against $k$
players who guess uniformly at random~\cite{drinen2009optimization}.
The target can be thought of as a random voter and
the guesses as candidate positions. The guessing game and plurality
winner position distributions are similar in shape, with two prominent bumps 
that move outward as $k$ grows; and both converge to uniform distributions. However, the point with the max
\emph{expected plurality vote share} (and max guessing game win
probability) is not quite the same as the point with the maximum
\emph{plurality win probability}, since a candidate's position
influences other candidates' vote shares.

\section{Non-uniform voters}\label{sec:non-uniform}
Given our understanding of the uniform voter case, we now broaden our scope and show that IRV exhibits exclusion zones more generally. We find that the same
``squeezing'' argument
can be applied to any symmetric voter distribution. The generalized
result hinges on a specific condition on the cumulative distribution function, \Cref{ineq:1/3-cond},
which intuitively
captures when, no matter where the last moderate candidate is, they
cannot be squeezed out by the most moderate extremists. This condition
is not always possible to satisfy non-trivially. After first giving
the general statement, we present special cases where the
condition is simple to state and satisfy---specifically, when the
voter density is monotonic over $[0, 1/2]$. If
the voter distribution is sufficiently highly polarized, the condition becomes
impossible to satisfy. In this \emph{hyper-polarized}
regime, the exclusion zone of IRV actually flips,
and IRV cannot elect moderate candidates over extreme ones. First, we present the general moderating effect of IRV for symmetric voter distributions. 

\begin{theorem}\label{thm:general-exclusion}
(General combinatorial moderation for IRV.) Let $f$ be symmetric over $[0, 1]$ with cdf $F$ and let $c \in (0, 1/2)$. If for all $x \in [c, 1/2]$,
 \begin{equation}\label{ineq:1/3-cond}
  F\left(\frac{x + 1-c}{2}\right) - F\left(\frac{c+ x}{2}\right)  > 1/3,
\end{equation}
then if there is at least one candidate in $[c, 1-c]$, the IRV winner must be in $[c, 1-c]$.
 \end{theorem}
\begin{proof}
  Suppose there is at least one candidate in $[0, c)$, at least one candidate in $(1-c, 1]$, and exactly one candidate $x$ in $[c, 1-c]$ (if there no candidates in the left or right extremes, then $x$ immediately wins by majority). Assume without loss of generality that $x \le 1/2$. Candidate $x$'s vote share is minimized when there are candidates at $c-\epsilon$ and $1-c + \epsilon$. The vote share of $x$ is then 
\begin{align*}
v(x) = F\left(\frac{x + 1-c + \epsilon}{2}\right) - F\left(\frac{c - \epsilon + x}{2}\right).
  \end{align*}
  
If Condition~(\ref{ineq:1/3-cond}) is satisfied, then
  then we can increase the left hand side of (\ref{ineq:1/3-cond}) to find
  \begin{align*}
    v(x) = F\left(\frac{x + 1-c + \epsilon}{2}\right) - F\left(\frac{c - \epsilon + x}{2}\right) > 1/3.
  \end{align*}
  Thus $x$ cannot be eliminated next, since there is a candidate with a smaller vote share than $x$. The IRV winner must therefore be in $[c, 1-c]$ by the same argument as in \Cref{thm:1/6}.
\end{proof}

We now consider two cases where Condition~(\ref{ineq:1/3-cond}) can be greatly simplified: when the voter distribution is moderate 
 ($f$ increases over $[0, 1/2]$; \Cref{thm:moderate-voters-exclusion}) and when voters are polarized ($f$ decreases over $[0, 1/2]$ but $F(1/4) < 1/3$; \Cref{thm:polarized-voters-exclusion}). The proofs in these cases follow the same structure, but differ in where moderate candidates are easiest to squeeze out (nearer or farther from 1/2). Proofs can be found in \Cref{app:proofs}. As another note, just as with \Cref{cor:probabilistic-moderation}, we immediately see from \Cref{thm:general-exclusion} (and the special cases below) that IRV has a probabilistic moderating effect with symmetric voter and candidate distributions (as long as they place positive mass on $[c, 1-c]$): as the number of candidates goes to infinity, the probability that the winner comes from $[c, 1-c]$ goes to 1.

\begin{theorem}\label{thm:moderate-voters-exclusion}
 (Moderate voter distribution.) Let $f$ be symmetric over $[0, 1]$ and non-decreasing over $[0, 1/2]$. For any $c \le F^{-1}(1/6)$, if there is a candidate in $[c, 1-c]$, then the IRV winner is in $[c, 1-c]$.
\end{theorem}

\begin{theorem}\label{thm:polarized-voters-exclusion}
  (Polarized voter distribution.) Let $f$ be symmetric over $[0, 1]$, non-increasing over $[0, 1/2]$, and let $F(1/4) < 1/3$. For any $c \le 2(F^{-1}(1/3) - 1/4)$, if there is a candidate in $[c, 1-c]$, then the IRV winner is in $[c, 1-c]$.
\end{theorem}

The uniform distribution is the unique distribution whose density function is both non-increasing and non-decreasing over $[0, 1/2]$. Indeed, for uniform $F(x) = x$, $1/6 = 2(F^{-1}(1/3) - 1/4) = F^{-1}(1/6)$. Note that for polarized voter distributions, \Cref{thm:polarized-voters-exclusion} requires $F(1/4) < 1/3$ (i.e., less than 1/3 of voters are left of 1/4). If the population is hyper-polarized and instead $F(1/4) > 1/3$, we can prove that IRV cannot elect moderates if both extremes are represented.

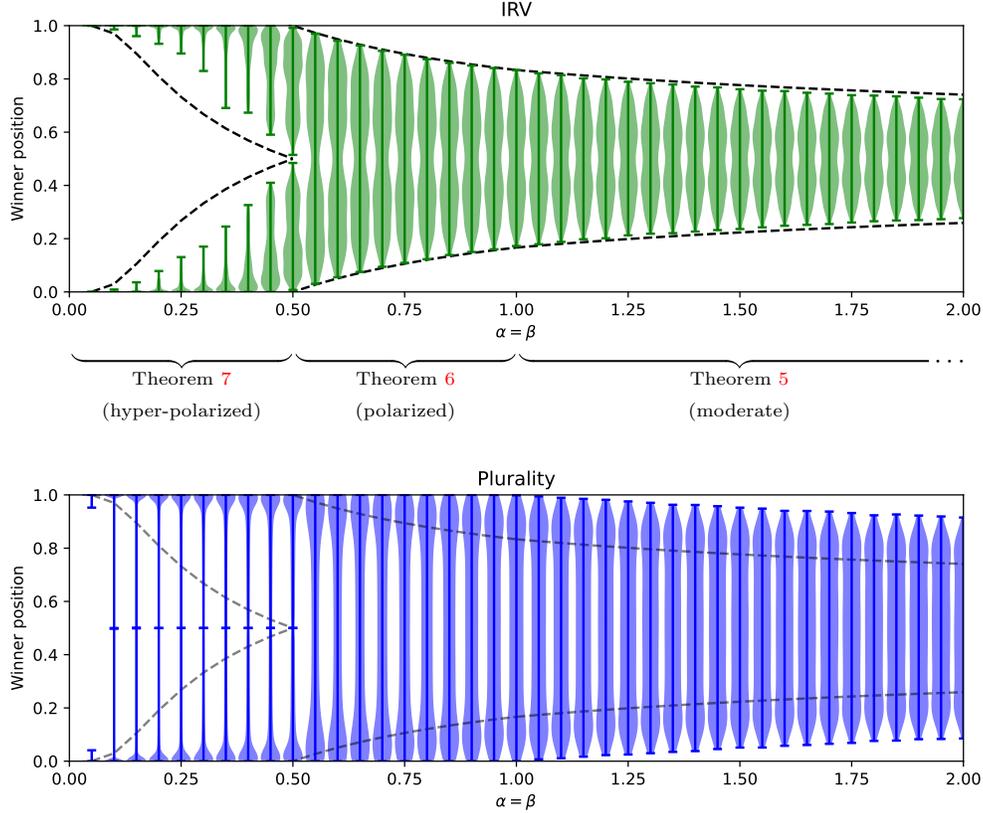
\begin{figure}[t]\centering
  \includegraphics[width=0.8\textwidth]{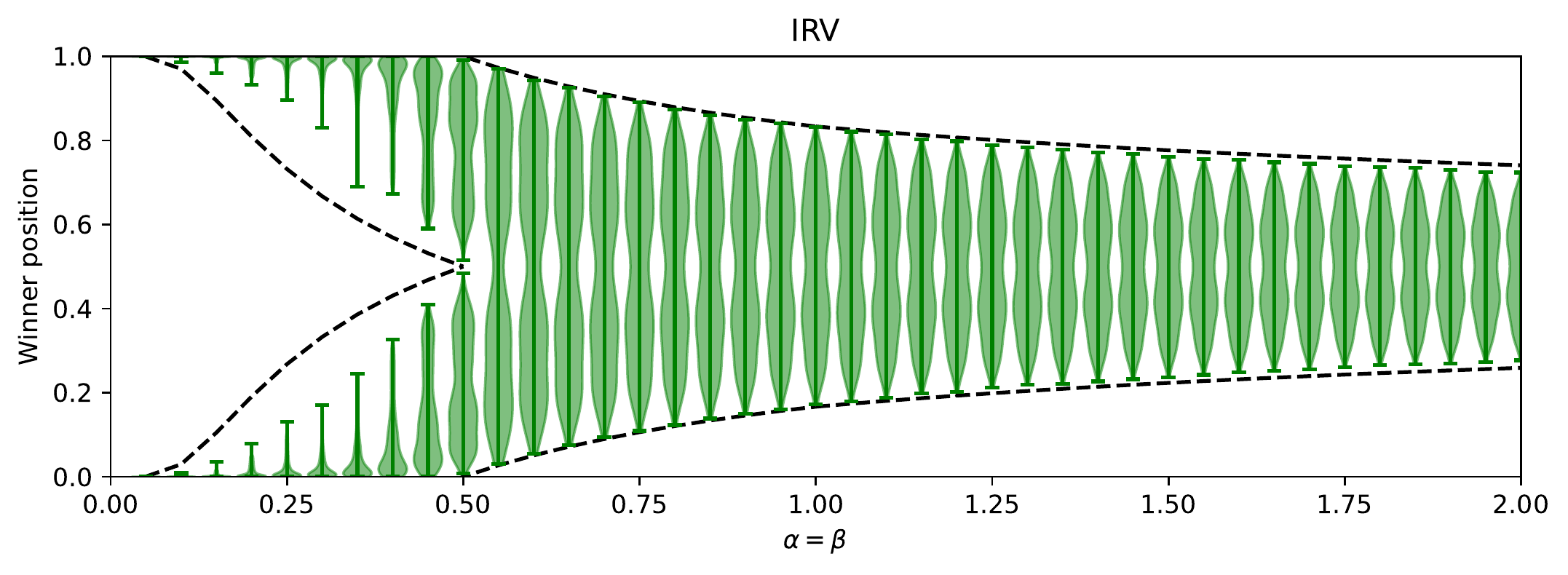}\\
  \begin{flushleft}
  \vspace{-5mm}
    \begin{tikzpicture}
    \node at (0, 0) {};
    \draw [decorate, 
    decoration = {calligraphic brace,amplitude=5pt}, line width=0.75pt,align=center] (5.42,0) --  (2.5,0) node [midway, below=3pt] {\scriptsize \Cref{thm:very-polarized-voters-exclusion} \\ \scriptsize (hyper-polarized)};
    \draw [decorate, 
    decoration = {calligraphic brace,amplitude=5pt}, line width=0.75pt,align=center] (8.39,0) --  (5.48,0) node [midway, below=3pt] {\scriptsize \Cref{thm:polarized-voters-exclusion} \\ \scriptsize(polarized)};
    \draw [decorate, 
    decoration = {calligraphic brace,amplitude=5pt}, line width=0.75pt,align=center] (14.3,0) --  (8.45,0) node [midway, below=3pt] {\scriptsize \Cref{thm:moderate-voters-exclusion} \\ \scriptsize(moderate)};
    \node[fill=white,inner sep=0pt, minimum height=2.5mm, minimum width=6mm] at (14.2, -.1) {$\cdots$};
    
%    \node[draw=none,fill=none] at (3.98,-.8){\includegraphics[width=7mm]{figures/beta-0.2.pdf}};
%    \node[draw=none,fill=none] at (6.96,-.8){\includegraphics[width=7mm]{figures/beta-0.7.pdf}};
%    \node[draw=none,fill=none] at (8.43,-.8){\includegraphics[width=7mm]{figures/beta-1.pdf}};    
%    \node[draw=none,fill=none] at (11.38,-.8){\includegraphics[width=7mm]{figures/beta-3.pdf}};
  \end{tikzpicture}\\
  \end{flushleft}
%  \vspace{-4mm}
  \includegraphics[width=0.8\textwidth]{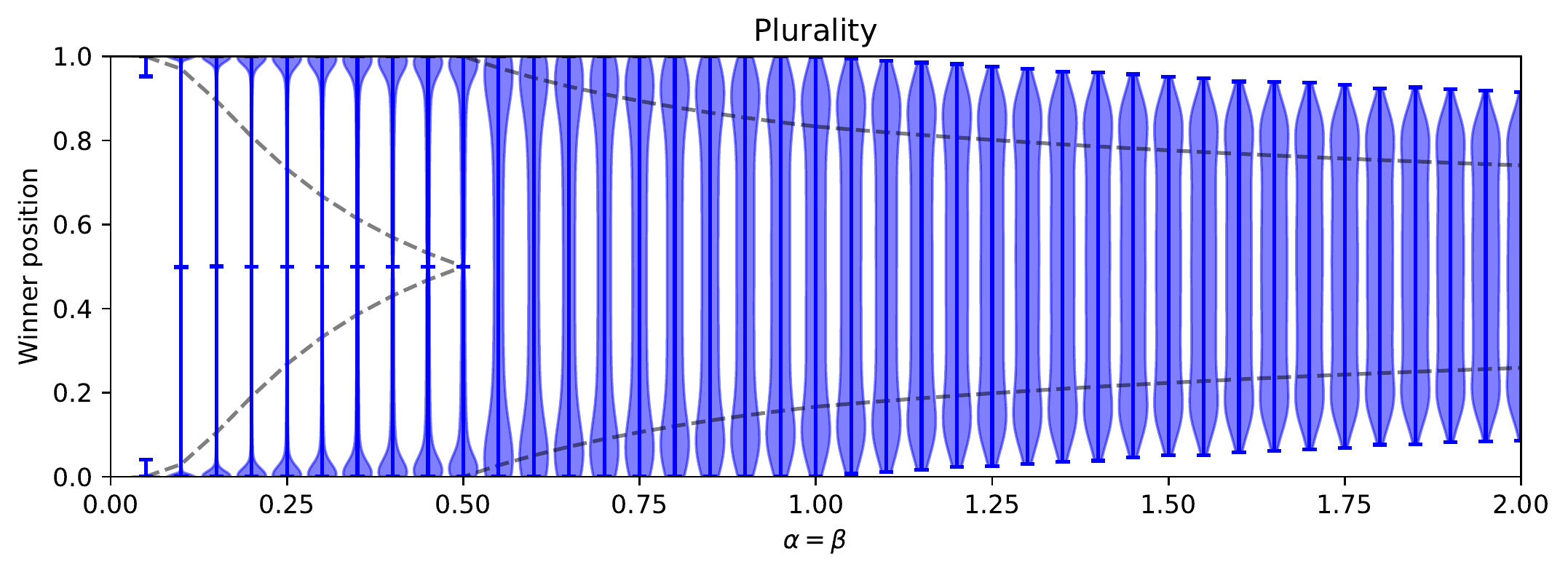}
  \caption{IRV (top) and plurality (bottom) winner positions with Beta$(\alpha, \alpha)$-distributed voters and candidates. The violin plots show empirical distributions from 100,000 simulation trials with $k=30$ candidates at each $\alpha$ value, with whiskers marking extrema. The dashed lines show the bounds from \Cref{thm:very-polarized-voters-exclusion,thm:polarized-voters-exclusion,thm:moderate-voters-exclusion} in the annotated ranges. As long as voters are not too polarized, IRV prevents extreme candidates from winning. Plurality, on the other hand, allows arbitrarily extreme candidates to win for $\alpha = 1$, when the voter distribution is uniform. }\label{fig:beta-sweep}
\end{figure}

\begin{theorem}\label{thm:very-polarized-voters-exclusion}
  (Hyper-polarized voter distribution.) Let $f$ be symmetric over $[0, 1]$ and let $F(1/4) > 1/3$. For any $c \ge 2 F^{-1}(1/3)$, if there is at least one candidate in $[0, c]$ and at least one candidate in $[1-c, 1]$, then the IRV winner must be in $[0, c]$ or $[1-c, 1]$. 
\end{theorem}

We saw in \Cref{thm:no-exclusion-uniform} that plurality has no exclusion zones for uniform voters. We now show that plurality has no exclusion zones regardless of the voter distribution (given mild continuity and positivity conditions), except the points 0 and 1. The proof can be found in \Cref{app:proofs}.

\begin{theorem}\label{thm:no-exclusion}
  (No combinatorial moderation for plurality.) Let $f$ be continuous and strictly positive over $(0, 1)$. Given any set of $\kappa \ge 1$ distinct candidate positions $x_1, \dots, x_\kappa$ with $x_1 \notin \{0, 1\}$, there exists a configuration of $k\ge \kappa$ candidates (including $x_1, \dots, x_\kappa$) such that the candidate at $x_1$ wins under plurality. If $x_1 \in\{0, 1\}$, then there exist voter distributions where $x_1$ cannot win under plurality.
\end{theorem}
 
\Cref{fig:beta-sweep} provides illustrations to accompany \Cref{thm:moderate-voters-exclusion,thm:polarized-voters-exclusion,thm:very-polarized-voters-exclusion,thm:no-exclusion}, showing empirical IRV and plurality winner positions when voters (and $k = 20$ candidates) are distributed according to symmetric Beta$(\alpha, \alpha)$ distributions. This family of Beta distributions is polarized for $\alpha < 1$, uniform for $\alpha = 1$, and moderate for $\alpha > 1$. \Cref{thm:moderate-voters-exclusion} thus applies for $\alpha \ge 1$. The crossover point between \Cref{thm:polarized-voters-exclusion} and \Cref{thm:very-polarized-voters-exclusion} occurs at $\alpha = 1/2$ (i.e., for Beta(1/2, 1/2), $F^{-1}(1/3) = 1/4$). \Cref{fig:beta-sweep} also visualizes the positions of plurality winners for these voter distributions, consistent with our analysis of plurality in~\Cref{thm:no-exclusion}.

 Finally, we revisit the  existing moderating effect result of \citeauthor{grofman2004if}  with single-peaked voters and strengthen it in the symmetric 1-Euclidean case. Recall that 1-Euclidean preferences are always single-peaked, but most sets of single-peaked preferences are not 1-Euclidean. That is, we make a stronger assumption on voter preferences and thus derive a stronger result. \citet{grofman2004if} proved that when voters have single-peaked preferences over $k \le 4$ candidates, if plurality elects the median candidate, so does IRV. The \emph{median candidate} here is defined as the candidate most preferred by the median voter (with single-peaked preferences, this is the Condorcet winner~\cite{black1948rationale}). With symmetric 1-Euclidean voters, the median candidate is the candidate closest to 1/2 (i.e., the most moderate candidate). Thus, applying the result of \citeauthor{grofman2004if} directly to the symmetric 1-Euclidean voter setting, we know for $k\le 4$ that whenever plurality elects the most moderate candidate, IRV does too. In the symmetric 1-Euclidean setting, we can strengthen this theorem to consider what happens when plurality does not elect the most moderate candidate. Note that this result holds for any symmetric voter distribution.

\begin{theorem}\label{thm:small-k-not-more-extreme}
For $k \le 4$ with symmetric 1-Euclidean voters, the IRV winner cannot be more extreme than the plurality winner (if no ties occur). For $k \ge 5$, the IRV winner can be more extreme than the plurality winner. 
\end{theorem}
\begin{proof}
The $k=1$ and $k=2$ cases are trivial, since IRV and plurality are identical when $k < 3$.

For $k=3$, suppose for a contradiction that the plurality winner $P$ is more moderate than the IRV winner $I$ (call the third candidate $E$). Under IRV, the first candidate eliminated can't be $I$ (since they win under IRV) and can't be $P$ (since they have the highest first-place vote share), so it must be $E$. In the second round of IRV, we are then left with a two-candidate plurality election between $I$ and $P$. Since voters are symmetrically distributed, the more moderate of $I$ and $P$ thus wins under IRV, which is $P$. Contradiction! 

For $k=4$, suppose again for a contradiction that the IRV winner $I$ is more extreme than the plurality winner $P$. As before, neither can be the first eliminated. Call the first  candidate eliminated $E$ and the fourth candidate $F$. Since $P$ is more moderate than $I$, the final IRV round cannot be between $P$ and $I$, or else $P$ would win, contradicting that $I$ is the IRV winner. Thus, the final round must be between $I$ and $F$. $P$ must then be the second eliminated after $E$. However, $P$ has a higher vote share than both $I$ and $F$ in the first round. To be eliminated second, the elimination of $E$ must cause $I$ and $F$ to overtake $P$. To redistribute votes to both $I$ and $F$, $E$ must be directly between them, with $P$ off to one side of the $I, E, F$ group. Consider two cases: (1) $P$ is adjacent to $I$. Since $P$ is more moderate than $I$, it must get all of the vote share on the side of opposite the $I, E, F$ group (either [1, 0.5] or [0.5, 1]), which means it has a majority---contradicting that $I$ is the IRV winner. (2) $P$ is adjacent to $F$. But then $F$ is more moderate than $I$, so $I$ cannot win in the final round---contradicting that it is the IRV winner.

For $k \ge 5$, place candidates at $\epsilon, 1/5, 1/2, 4/5,$ and $1$ for small $\epsilon$ (for instance $\epsilon\le 0.01$ works; additional candidates can be packed into $[0, \epsilon]$). Note that the candidate at $1/2$ is the plurality winner, with vote share $3/10$. The candidates in $[0, \epsilon]$ are eliminated first under IRV, followed by the candidates at $1$ and $\epsilon$. At this point, the candidates at $1/5$ and $4/5$ have a higher vote share than the candidate at $1/2$, who is eliminated. The IRV winner is then either at $1/5$ or $4/5$.
\end{proof}

See \Cref{fig:plurality-irv-scatter} in \Cref{app:figures} for simulation results demonstrating \Cref{thm:small-k-not-more-extreme}. All simulation code and results are available at \url{https://github.com/tomlinsonk/irv-moderation}.

\section{Discussion}

We began by considering a contrast between IRV and plurality voting
when the positions of voters and candidates are drawn from the uniform 
distribution on the unit interval:
in this case, IRV (unlike plurality)
has a moderating effect, with the probability 
that the winner comes from the interval $[1/6,5/6]$ converging to 1
as the number of candidates goes to infinity.
This moderating effect continues to hold (with proper sub-intervals
different from $[1/6,5/6]$) even as the distribution
of voters and candidates becomes more polarized, with an increasing
amount of probability mass near the endpoints of the interval,
until a specific threshold of hyper-polarization is reached.
Our analysis also provides methods for determining the exact distribution
of winner positions in certain cases, making more fine-grained
comparisons between IRV and plurality possible. 

It would be interesting to consider extensions of our work in a number of
directions, and here we highlight three of these.
First, we did not consider
strategic analyses (e.g., of Nash equilibria, as in \citet{dellis2017policy}), and were instead
motivated by bounded rationality~\cite{bendor2011behavioral} and a need to better understand the
underlying voting system, focusing on the non-strategic setting where
candidate positions are fixed. 
For instance,
how might candidates behave strategically given an understanding 
of IRV exclusion zones or 
the winner position distribution of IRV? Behavioral evidence for
bounded rationality indicates that people tend to operate at a low
strategic
depth~\cite{stahl1995players,colman2003depth,ohtsubo2006depth}. In
this framework, level-0 players act randomly, level-1 players
calculate best responses to level-0 players, and so on. Our analysis
therefore corresponds to level-0 strategic reasoning, and can be used
as a starting point for analysis of higher-order strategy. 

Second, we modeled voting populations
as symmetric continuous distributions in one dimension, with
preferences arising strictly from distances in this dimension. 
Considering
higher-dimensional preference spaces would also be a natural extension
of our analysis. Does IRV exhibit exclusion zones in two, three, or
more dimensions? Asymmetric voter distributions would also be valuable
to consider, although the notion of a \emph{moderate} may need to be
revisited in this case (perhaps based on the median voter). Using the
same squeezing argument, IRV should also exhibit exclusion zones with
asymmetric voter distributions, although their forms may not be as
tidy as the ones we derive. Other possible extensions include non-linear voter preferences (for instance, where a voter ranks all candidates on their right before all candidates on their left, regardless of distance), probabilistic voting, and voter abstention. Practical considerations of IRV could also be taken into account; for instance, real-world elections often ask for top-truncated preferences rather than full rankings, which can the affect the outcome~\cite{tomlinson2023ballot}. Does IRV with truncated ballots still exhibit a moderating effect?

Finally, as we noted earlier, there are voting systems that always
select the most moderate candidate with symmetric 1-Euclidean voters.
This is true for any system that satisfies the Condorcet criterion,
selecting the Condorcet winner whenever one exists (a property that holds for the minimax, Condorcet-Hare, Copeland, and Dodgson 
methods, among many others~\cite{black1958theory,richelson1975comparative,green2016statistical}); it is also true for
some other voting systems that do not in general satisfy the Condorcet
criterion, like the Coombs rule~\cite{coombs1964theory,grofman2004if}.
There are a variety of practical and historical reasons why these
methods are not widely used for political elections. For instance,
Dodgson's method is NP-hard to compute~\cite{bartholdi1989voting} and the Coombs rule is very
sensitive to incomplete ballots, which are common in practice. As we
are motivated by ongoing debates about IRV and plurality, our
attention has been restricted to these voting methods. However, a
broader understanding of moderating effects of voting systems would be
valuable. There has been some theoretical work on moderating effects
of score-based voting systems (like Borda count and approval voting)
with strategic voters and candidates~\cite{dellis2009would}. However,
it is an open question (with some computational evidence to support it
\cite{chamberlin1978toward}) whether other voting systems like Borda count
exert a moderating effect in the setting we study, with fixed voter
and candidate distributions.

\section*{Acknowledgments}
This work was supported in part by ARO MURI, a Simons Investigator Award, a Simons Collaboration grant, a grant from the MacArthur Foundation, the Koret Foundation, and NSF CAREER Award \#2143176. We thank Robert Kleinberg and Spencer Peters for suggesting the circle-cutting argument used to prove \Cref{thm:plurality-asymptotic-uniform}.

\bibliographystyle{plain}
\bibliography{references}

\begin{thebibliography}{54}
\providecommand{\natexlab}[1]{#1}
\providecommand{\url}[1]{\texttt{#1}}
\expandafter\ifx\csname urlstyle\endcsname\relax
  \providecommand{\doi}[1]{doi: #1}\else
  \providecommand{\doi}{doi: \begingroup \urlstyle{rm}\Url}\fi

\bibitem[Atkinson and Ganz(2022)]{atkinson}
N.~Atkinson and S.~C. Ganz.
\newblock The flaw in ranked-choice voting: rewarding extremists.
\newblock \emph{The Hill}, 2022.

\bibitem[Aziz et~al.(2017)Aziz, Brill, Conitzer, Elkind, Freeman, and Walsh]{aziz2017justified}
H.~Aziz, M.~Brill, V.~Conitzer, E.~Elkind, R.~Freeman, and T.~Walsh.
\newblock Justified representation in approval-based committee voting.
\newblock \emph{Social Choice and Welfare}, 48\penalty0 (2):\penalty0 461--485, 2017.

\bibitem[Bartholdi et~al.(1989)Bartholdi, Tovey, and Trick]{bartholdi1989voting}
J.~Bartholdi, C.~A. Tovey, and M.~A. Trick.
\newblock Voting schemes for which it can be difficult to tell who won the election.
\newblock \emph{Social Choice and welfare}, 6:\penalty0 157--165, 1989.

\bibitem[Bendor et~al.(2011)Bendor, Diermeier, Siegel, and Ting]{bendor2011behavioral}
J.~Bendor, D.~Diermeier, D.~A. Siegel, and M.~Ting.
\newblock A behavioral theory of elections.
\newblock In \emph{A Behavioral Theory of Elections}. Princeton University Press, 2011.

\bibitem[Black(1948)]{black1948rationale}
D.~Black.
\newblock On the rationale of group decision-making.
\newblock \emph{Journal of Political Economy}, 56\penalty0 (1):\penalty0 23--34, 1948.

\bibitem[Black(1958)]{black1958theory}
D.~Black.
\newblock \emph{The theory of committees and elections}.
\newblock Springer, 1958.

\bibitem[Bogomolnaia and Laslier(2007)]{bogomolnaia2007euclidean}
A.~Bogomolnaia and J.-F. Laslier.
\newblock Euclidean preferences.
\newblock \emph{Journal of Mathematical Economics}, 43\penalty0 (2):\penalty0 87--98, 2007.

\bibitem[Bordignon et~al.(2016)Bordignon, Nannicini, and Tabellini]{bordignon2016moderating}
M.~Bordignon, T.~Nannicini, and G.~Tabellini.
\newblock Moderating political extremism: single round versus runoff elections under plurality rule.
\newblock \emph{American Economic Review}, 106\penalty0 (8):\penalty0 2349--70, 2016.

\bibitem[Boutilier et~al.(2012)Boutilier, Caragiannis, Haber, Lu, Procaccia, and Sheffet]{boutilier2012optimal}
C.~Boutilier, I.~Caragiannis, S.~Haber, T.~Lu, A.~D. Procaccia, and O.~Sheffet.
\newblock Optimal social choice functions: A utilitarian view.
\newblock In \emph{Proceedings of the 13th ACM Conference on Electronic Commerce}, pages 197--214, 2012.

\bibitem[Brill et~al.(2022)Brill, Israel, Micha, and Peters]{brill2022individual}
M.~Brill, J.~Israel, E.~Micha, and J.~Peters.
\newblock Individual representation in approval-based committee voting.
\newblock In \emph{Proceedings of the AAAI Conference on Artificial Intelligence}, volume~36, pages 4892--4899, 2022.

\bibitem[Chamberlin and Cohen(1978)]{chamberlin1978toward}
J.~R. Chamberlin and M.~D. Cohen.
\newblock Toward applicable social choice theory: A comparison of social choice functions under spatial model assumptions.
\newblock \emph{American Political Science Review}, 72\penalty0 (4):\penalty0 1341--1356, 1978.

\bibitem[Colman(2003)]{colman2003depth}
A.~M. Colman.
\newblock Depth of strategic reasoning in games.
\newblock \emph{Trends in Cognitive Sciences}, 7\penalty0 (1):\penalty0 2--4, 2003.

\bibitem[Coombs(1964)]{coombs1964theory}
C.~H. Coombs.
\newblock \emph{A theory of data}.
\newblock Wiley, 1964.

\bibitem[Darling(1953)]{darling1953class}
D.~A. Darling.
\newblock On a class of problems related to the random division of an interval.
\newblock \emph{The Annals of Mathematical Statistics}, pages 239--253, 1953.

\bibitem[Dean(October 7 2016)]{dean2016}
H.~Dean.
\newblock How to move beyond the two-party system.
\newblock \emph{The New York Times}, October 7 2016.

\bibitem[DellaPosta et~al.(2015)DellaPosta, Shi, and Macy]{dellaposta2015liberals}
D.~DellaPosta, Y.~Shi, and M.~Macy.
\newblock Why do liberals drink lattes?
\newblock \emph{American Journal of Sociology}, 120\penalty0 (5):\penalty0 1473--1511, 2015.

\bibitem[Dellis(2009)]{dellis2009would}
A.~Dellis.
\newblock Would letting people vote for multiple candidates yield policy moderation?
\newblock \emph{Journal of Economic Theory}, 144\penalty0 (2):\penalty0 772--801, 2009.

\bibitem[Dellis et~al.(2017)Dellis, Gauthier-Belzile, and Oak]{dellis2017policy}
A.~Dellis, A.~Gauthier-Belzile, and M.~Oak.
\newblock Policy polarization and strategic candidacy in elections under the alternative-vote rule.
\newblock \emph{Journal of Institutional and Theoretical Economics}, pages 565--590, 2017.

\bibitem[DeMarzo et~al.(2003)DeMarzo, Vayanos, and Zwiebel]{demarzo2003persuasion}
P.~M. DeMarzo, D.~Vayanos, and J.~Zwiebel.
\newblock Persuasion bias, social influence, and unidimensional opinions.
\newblock \emph{The Quarterly Journal of Economics}, 118\penalty0 (3):\penalty0 909--968, 2003.

\bibitem[Diamond(October 13 2016)]{diamond2016}
L.~Diamond.
\newblock The second-most important vote on {Nov}.\ 8.
\newblock \emph{Foreign Policy}, October 13 2016.

\bibitem[Donovan et~al.(2016)Donovan, Tolbert, and Gracey]{donovan2016campaign}
T.~Donovan, C.~Tolbert, and K.~Gracey.
\newblock Campaign civility under preferential and plurality voting.
\newblock \emph{Electoral Studies}, 42:\penalty0 157--163, 2016.

\bibitem[Downs(1957)]{downs1957economic}
A.~Downs.
\newblock \emph{An economic theory of democracy}.
\newblock Harper \& Row, 1957.

\bibitem[Drinen et~al.(2009)Drinen, Kennedy, and Priestley]{drinen2009optimization}
D.~Drinen, K.~G. Kennedy, and W.~M. Priestley.
\newblock An optimization problem with a surprisingly simple solution.
\newblock \emph{The American Mathematical Monthly}, 116\penalty0 (4):\penalty0 328--341, 2009.

\bibitem[Ebadian et~al.(2022)Ebadian, Kahng, Peters, and Shah]{ebadian2022optimized}
S.~Ebadian, A.~Kahng, D.~Peters, and N.~Shah.
\newblock Optimized distortion and proportional fairness in voting.
\newblock In \emph{Proceedings of the 23rd ACM Conference on Economics and Computation}, pages 563--600, 2022.

\bibitem[Elkind et~al.()Elkind, Lackner, and Peters]{elkind2022restrictions}
E.~Elkind, M.~Lackner, and D.~Peters.
\newblock Preference restrictions in computational social choice: A survey.
\newblock \emph{arXiv preprint: arXiv:2205.09092}.
\newblock URL \url{https://arxiv.org/abs/2205.09092}.

\bibitem[Fraenkel and Grofman(2004)]{fraenkel2004neo}
J.~Fraenkel and B.~Grofman.
\newblock A neo-{Downsian} model of the alternative vote as a mechanism for mitigating ethnic conflict in plural societies.
\newblock \emph{Public Choice}, pages 487--506, 2004.

\bibitem[Fraenkel and Grofman(2006{\natexlab{a}})]{fraenkel2006does}
J.~Fraenkel and B.~Grofman.
\newblock Does the alternative vote foster moderation in ethnically divided societies? the case of {Fiji}.
\newblock \emph{Comparative Political Studies}, 39\penalty0 (5):\penalty0 623--651, 2006{\natexlab{a}}.

\bibitem[Fraenkel and Grofman(2006{\natexlab{b}})]{fraenkel2006failure}
J.~Fraenkel and B.~Grofman.
\newblock The failure of the alternative vote as a tool for ethnic moderation in {Fiji}: A rejoinder to {Horowitz}.
\newblock \emph{Comparative Political Studies}, 39\penalty0 (5):\penalty0 663--666, 2006{\natexlab{b}}.

\bibitem[Gkatzelis et~al.(2020)Gkatzelis, Halpern, and Shah]{gkatzelis2020resolving}
V.~Gkatzelis, D.~Halpern, and N.~Shah.
\newblock Resolving the optimal metric distortion conjecture.
\newblock In \emph{2020 IEEE 61st Annual Symposium on Foundations of Computer Science (FOCS)}, pages 1427--1438. IEEE, 2020.

\bibitem[Green-Armytage et~al.(2016)Green-Armytage, Tideman, and Cosman]{green2016statistical}
J.~Green-Armytage, T.~N. Tideman, and R.~Cosman.
\newblock Statistical evaluation of voting rules.
\newblock \emph{Social Choice and Welfare}, 46:\penalty0 183--212, 2016.

\bibitem[Grofman and Feld(2004)]{grofman2004if}
B.~Grofman and S.~L. Feld.
\newblock If you like the alternative vote (aka the instant runoff), then you ought to know about the {Coombs} rule.
\newblock \emph{Electoral Studies}, 23\penalty0 (4):\penalty0 641--659, 2004.

\bibitem[Halpern et~al.(2023)Halpern, Kehne, Procaccia, Tucker-Foltz, and W{\"u}thrich]{halpern2023representation}
D.~Halpern, G.~Kehne, A.~D. Procaccia, J.~Tucker-Foltz, and M.~W{\"u}thrich.
\newblock Representation with incomplete votes.
\newblock In \emph{Proceedings of the AAAI Conference on Artificial Intelligence}, volume~37, pages 5657--5664, 2023.

\bibitem[Holst(1980)]{holst1980lengths}
L.~Holst.
\newblock On the lengths of the pieces of a stick broken at random.
\newblock \emph{Journal of Applied Probability}, 17\penalty0 (3):\penalty0 623--634, 1980.

\bibitem[Horowitz(2006)]{horowitz2006strategy}
D.~L. Horowitz.
\newblock Strategy takes a holiday: {Fraenkel and Grofman} on the alternative vote.
\newblock \emph{Comparative Political Studies}, 39\penalty0 (5):\penalty0 652--662, 2006.

\bibitem[Horowitz(2007)]{horowitz2007have}
D.~L. Horowitz.
\newblock Where have all the parties gone? {Fraenkel and Grofman} on the alternative vote--yet again.
\newblock \emph{Public Choice}, 133\penalty0 (1):\penalty0 13--23, 2007.

\bibitem[Hotelling(1929)]{hotelling1929stability}
H.~Hotelling.
\newblock Stability in competition.
\newblock \emph{The Economic Journal}, 39\penalty0 (153):\penalty0 41--57, 1929.

\bibitem[John and Douglas(2017)]{john2017candidate}
S.~John and A.~Douglas.
\newblock Candidate civility and voter engagement in seven cities with ranked choice voting.
\newblock \emph{National Civic Review}, 106\penalty0 (1):\penalty0 25--29, 2017.

\bibitem[Kahng et~al.(2023)Kahng, Latifian, and Shah]{kahng2023voting}
A.~Kahng, M.~Latifian, and N.~Shah.
\newblock Voting with preference intensities.
\newblock In \emph{Proceedings of the AAAI Conference on Artificial Intelligence}, volume~37, pages 5697--5704, 2023.

\bibitem[Kropf(2021)]{kropf2021using}
M.~Kropf.
\newblock Using campaign communications to analyze civility in ranked choice voting elections.
\newblock \emph{Politics and Governance}, 9\penalty0 (2):\penalty0 280--292, 2021.

\bibitem[Layman et~al.(2006)Layman, Carsey, and Horowitz]{layman2006party}
G.~C. Layman, T.~M. Carsey, and J.~M. Horowitz.
\newblock Party polarization in american politics.
\newblock \emph{Annual Review of Political Science}, 9:\penalty0 83--110, 2006.

\bibitem[McGann et~al.(2002)McGann, Grofman, and Koetzle]{mcgann2002party}
A.~J. McGann, B.~Grofman, and W.~Koetzle.
\newblock Why party leaders are more extreme than their members: Modeling sequential elimination elections in the {US House of Representatives}.
\newblock \emph{Public Choice}, 113\penalty0 (3-4):\penalty0 337--356, 2002.

\bibitem[Merrill~III(1984)]{merrill1984comparison}
S.~Merrill~III.
\newblock A comparison of efficiency of multicandidate electoral systems.
\newblock \emph{American Journal of Political Science}, pages 23--48, 1984.

\bibitem[Mitchell(2014)]{mitchell2014single}
P.~Mitchell.
\newblock The single transferable vote and ethnic conflict: the evidence from {Northern Ireland}.
\newblock \emph{Electoral Studies}, 33:\penalty0 246--257, 2014.

\bibitem[Ohtsubo and Rapoport(2006)]{ohtsubo2006depth}
Y.~Ohtsubo and A.~Rapoport.
\newblock Depth of reasoning in strategic form games.
\newblock \emph{The Journal of Socio-Economics}, 35\penalty0 (1):\penalty0 31--47, 2006.

\bibitem[Osborne(1995)]{osborne1995spatial}
M.~J. Osborne.
\newblock Spatial models of political competition under plurality rule: A survey of some explanations of the number of candidates and the positions they take.
\newblock \emph{Canadian Journal of economics}, pages 261--301, 1995.

\bibitem[Poole and Rosenthal(1984)]{poole1984polarization}
K.~T. Poole and H.~Rosenthal.
\newblock The polarization of american politics.
\newblock \emph{The Journal of Politics}, 46\penalty0 (4):\penalty0 1061--1079, 1984.

\bibitem[Poole and Rosenthal(1991)]{poole1991patterns}
K.~T. Poole and H.~Rosenthal.
\newblock Patterns of congressional voting.
\newblock \emph{American Journal of Political Science}, pages 228--278, 1991.

\bibitem[Reilly(2018)]{reilly2018centripetalism}
B.~Reilly.
\newblock Centripetalism and electoral moderation in established democracies.
\newblock \emph{Nationalism and Ethnic Politics}, 24\penalty0 (2):\penalty0 201--221, 2018.

\bibitem[Richelson(1975)]{richelson1975comparative}
J.~Richelson.
\newblock A comparative analysis of social choice functions.
\newblock \emph{Behavioral science}, 20\penalty0 (5):\penalty0 331--337, 1975.

\bibitem[Slaughter(2019)]{slaughter}
A.-M. Slaughter.
\newblock How to fix polarization: Ranked choice voting.
\newblock Politico, 2019.

\bibitem[Stahl and Wilson(1995)]{stahl1995players}
D.~O. Stahl and P.~W. Wilson.
\newblock On players' models of other players: Theory and experimental evidence.
\newblock \emph{Games and Economic Behavior}, 10\penalty0 (1):\penalty0 218--254, 1995.

\bibitem[Tomlinson et~al.(2023)Tomlinson, Ugander, and Kleinberg]{tomlinson2023ballot}
K.~Tomlinson, J.~Ugander, and J.~Kleinberg.
\newblock Ballot length in instant runoff voting.
\newblock In \emph{Proceedings of the 37th AAAI Conference on Artificial Intelligence}, 2023.

\bibitem[Wang et~al.(2021)Wang, Cervas, Grofman, and Lipsitz]{wang2021systems}
S.~S.-H. Wang, J.~Cervas, B.~Grofman, and K.~Lipsitz.
\newblock A systems framework for remedying dysfunction in us democracy.
\newblock \emph{Proceedings of the National Academy of Sciences}, 118\penalty0 (50):\penalty0 e2102154118, 2021.

\bibitem[Waxman(November 3 2016)]{waxman2016}
S.~Waxman.
\newblock Ranked-choice voting is not the solution.
\newblock \emph{Democracy}, November 3 2016.

\end{thebibliography}
\clearpage

\appendix

\section{Quotes about IRV Moderation}\label{app:quotes}

In this section, we include some quotes indicating that moderating effects of IRV are an often-invoked argument in policy debates, suggesting the existence of a folk theory which we formalize in this work. Recall that IRV is commonly referred to as ranked-choice voting in United States.

\begin{table}[h]
\centering
\caption{Quotes for and against a moderating effect of IRV}
  \begin{tabular}{p{15cm}}
  \toprule
    
  \textit{[Under ranked-choice voting,] voters can support their favorites while still voting effectively against their least favorite. Having more competition encourages better dialogue on issues. Civility is substantially improved. Needing to reach out to more voters leads candidates to reduce personal attacks and govern more inclusively.}\\
  ---Howard Dean, former Governor of Vermont~\cite{dean2016}\\[1em]
  
  \textit{We need an electoral system that breaks the current stranglehold of the two-party monopoly, one that would allow voters to choose between a much more nuanced range of positions than “extreme” versus “moderate,” would allow third-party candidates to run without being spoilers and would encourage more civil campaigning and political discourse. The solution is to adopt ranked-choice voting for all state and federal elections [....] We need [...] reforms that will allow the American people to reassert our power over a party system that is badly broken and compel candidates to appeal to a far broader swathe of us than a narrow ``base.''}\\---Anne-Marie Slaughter, CEO of New America~\cite{slaughter}\\[1em]

  \textit{Quite to the contrary, [ranked-choice voting] may give life to more strident candidates, hoping to siphon first-place ballots from extreme voters who will give second preference to whichever major party is closest to them. This could result in more comity between the major-party candidates, as fringier competitors blot the airwaves with attacks. Or it might produce strategic coalitions sniping at each other, leaving us effectively back where we started.}\\
  ---Simon Waxman, former managing editor at Boston Review~\cite{waxman2016}\\[1em]

  \textit{However, ranked-choice voting makes it more difficult to elect moderate candidates when the electorate is polarized. For example, in a three-person race, the moderate candidate may be preferred to each of the more extreme candidates by a majority of voters. However, voters with far-left and far-right views will rank the candidate in second place rather than in first place. Since ranked-choice voting counts only the number of first-choice votes (among the remaining candidates), the moderate candidate would be eliminated in the first round, leaving one of the extreme candidates to be declared the winner. [...] The ranked-choice system that is being used around the country to conduct elections with more than two candidates is biased towards extreme candidates and away from moderate ones.}\\
  ---Nathan Atkinson, Assistant Professor at University of Wisconsin Law School, and Scott C. Ganz, Associate Teaching Professor at Georgetown University)~\cite{atkinson}\\
  
  \bottomrule
  \end{tabular}
\end{table}

\section{Additional proofs}\label{app:proofs}
Since \Cref{thm:plurality-asymptotic-uniform} relies on \Cref{lemma:gumbel}, we first prove \Cref{lemma:gumbel}.

\begin{proof}[Proof of \Cref{lemma:gumbel}]
For notational simplicity, define $n = k+1$ to be the number of gaps between candidates (including the leftmost and rightmost gaps bounded by 0 and 1) and let $S_1, \dots, S_{n}$ be the sizes of the gaps. Additionally, let $X_1, \dots, X_n$ be i.i.d.\ exponential random variables with mean 1 and let $T_n = \sum_{i=1}^n X_i$ be their sum. We then have $S_i = X_i / T_n$~\cite{holst1980lengths} and $T_n$ is independent of each $S_i$ (the $S_i$ are not independent, however). Since voters are uniform, the vote shares are $S_1 + S_{2}/2$ for the leftmost candidate and $(S_i + S_{i+1})/2$ for the $i$th candidate for $i=2, \dots, n-2$. For the rightmost candidate, we introduce an alternative indexing to avoid a subscript dependent on $n$. Let $S_r$ and $S_{r'}$ be the rightmost and second-rightmost gaps, respectively, so that the rightmost candidate's vote share is $S_{r'}/2 + S_r$ (ditto for $X_r$ and $X_r'$).  We then have 
\begin{align*}
S_1 + S_{2} / 2 &= \frac{1}{T_n}(X_1 + X_2 / 2),\\
(S_i + S_{i+1})/2  &= \frac{1}{T_n}(X_i + X_{i+1})/2,\\
S_{r'} / 2 + S_{r}  &= \frac{1}{T_n}(X_{r'} / 2 + X_{r}).
\end{align*}

Consider the asymptotic CDF of the leftmost candidate's vote share scaled by $n \ge 1$:
\begin{align*}
  \lim_{n\rightarrow \infty}\Pr(n(S_1 + S_{2} / 2) \le z) &= \lim_{n\rightarrow \infty} \Pr\left(\frac{n}{T_n}(X_1 + X_2 / 2) \le z\right)\\
  &= \lim_{n\rightarrow \infty} \Pr\left(X_1 + X_2 / 2 \le \frac{T_n}{n} z\right).
\end{align*}
Since it is the sum of $n$ independent exponential RVs with mean $1$, $T_n$ has a Gamma$(n, 1)$ distribution~\cite{holst1980lengths}, so it has variance $n$ and expectation $n$. Thus, $\Var(T_n / n) = 1/n^2 \Var(T_n) = 1/n$. Since $\Var(T_n / n) \rightarrow 0$ as $n \rightarrow \infty$ and $\E[T_n / n] = 1$, we have $\lim_{n \rightarrow \infty} T_n / n = 1$. This can also be seen using the Law of Large Numbers, since $\sum_i X_i /n$ converges in probability to $\E[X_i] = 1$.  Thus, 
\begin{align*}
  \lim_{n\rightarrow \infty} \Pr\left(X_1 + X_2 / 2 \le \frac{T_n}{n} z\right) 
  =
  \Pr\left(X_1 + X_2 / 2 \le z\right).
\end{align*}
Similarly for middle candidates:
\begin{align*}
\lim_{n\rightarrow \infty}\Pr(n(S_{i} + S_{i+1}) / 2 \le z) &= \lim_{n\rightarrow \infty} \Pr\left(\frac{n}{T_n}(X_{i} + X_{i+1}) / 2 \le z\right)\\
  &= \lim_{n\rightarrow \infty} \Pr\left((X_{i} + X_{i+1}) / 2 \le \frac{T_n}{n} z\right)\\
&= \Pr\left((X_{i} + X_{i+1}) / 2 \le z\right).
\end{align*}
Likewise, for the rightmost candidate,
\begin{align*}
\lim_{n\rightarrow \infty}\Pr(n(S_{r'}/2 + S_{r}) \le z) &= \lim_{n\rightarrow \infty} \Pr\left(\frac{n}{T_n}(X_{r'}/2 + X_r) \le z\right)\\
  &= \lim_{n\rightarrow \infty} \Pr\left(X_{r'}/2 + X_r \le \frac{T_n}{n} z\right)\\
&= \Pr\left(X_{r'}/2 + X_r \le  z\right).
\end{align*}

Thus, the asymptotic distributions of $n$ times the vote shares equal the distributions of the corresponding sums of exponentials RVs. This is the same idea used in proving the asymptotic distribution of the maximum gap size~\cite{holst1980lengths}. Now consider the distribution of the maximum vote share. Let $V_k$ be the maximum vote share with $k\ge 3$ candidates (and therefore $n=k+1$ gaps between candidates) and let $M_k$ be the maximum corresponding exponential RV sum. As above, by LLN, $\lim_{n\rightarrow\infty}\Pr(nV_k \le z) = \Pr(M_k \le z)$. Let $L_n = X_1 + X_2 / 2$, $C_{i} = (X_i + X_{i+1})/2$, and $R_n = X_{n-1}/ 2 + X_n$.

\begin{align*}
\Pr(M_k \le z) &= \Pr(L_n \le z, C_2 \le z, \dots, C_{n-2},\le z, R_n \le z)\\
&= \Pr(L_n \le z) \Pr(R_n \le z \mid L_n \le z, C_2 \le z, \dots, C_{n-2}\le z)\\
& \qquad \qquad \cdot  \prod_{i = 2}^{n-2} \Pr(C_i \le z \mid L_n \le z, \dots, C_{i-1} \le z).
\end{align*}
Then, using the facts that each $X_i$ is independent (and thus $C_i$ is independent of $C_{j}$ for $j > i+1$ and $j < i - 1$) and that each $C_i$ is identically distributed, we can simplify the conditioning:
\begin{align*}
  \Pr(M_k \le z) &=  \Pr(L_n \le z) \Pr(C_2 \le z \mid L_n \le z) \Pr(R_n \le z \mid C_{n-2}\le z)\prod_{i = 3}^{n-2} \Pr(C_i \le z \mid  C_{i-1} \le z)\\
  &=  \Pr(L_n \le z) \Pr(C_2 \le z \mid L_n \le z) \Pr(R_n \le z \mid C_{n-2}\le z) \Pr(C_i \le z \mid  C_{i-1} \le z)^{n-4}.
\end{align*}
We now only have four different probabilities to compute.

\begin{enumerate}
  \item $\Pr(L_n \le z) = \Pr(X_1 + X_2/2 \le z)$. Note that $X_1 \sim \text{Exp}(1)$, $X_2/2 \sim \text{Exp}(2)$, and $X_1$ and $X_2/2$ are independent. Thus:
  \begin{align*}
    \Pr(X_1 + X_2/2 \le z) &= \int_0^z e^{-t}\left(\int_0^{z-t} 2e^{-2s}ds \right) dt\\
    &= 1-e^{-2z}-2e^{-z}.
  \end{align*}

  \item $$\Pr(C_2 \le z \mid L_n \le z) = \frac{\Pr((X_2+X_3) /2 \le z, X_1 + X_2/2 \le z)}{\Pr(X_1 + X_2/2 \le z)}.$$
  We already know the denominator from the previous calculation. As before, $X_2/2$ and $X_3/2$ are independent and $\text{Exp}(2)$ distributed. To have $(X_2+X_3) /2 \le z$ and $X_1 + X_2/2 \le z$, we first need $X_1 \le z$, then $X_2/2 \le z- X_1$, and finally $X_3/2 \le z - X_2/2$. Thus:
  \begin{align*}
    \Pr((X_2+X_3) /2 \le z, X_1 + X_2/2 \le z) &= \int_{0}^z e^{-t} \left(\int_{0}^{z-t} 2e^{-2s} \left[\int_{0}^{z-s} 2e^{-2r} \,dr\right]\, ds\right) \,dt\\
    &= 1 - 2 e^{-3 z} + 3 e^{-2 z} - 2 e^{-z} - 2 e^{-2 z} z.
  \end{align*}
  So,
  \begin{align*}
    \Pr(C_2 \le z \mid L_n \le z) &= \frac{1 - 2 e^{-3 z} + 3 e^{-2 z} - 2 e^{-z} - 2 e^{-2 z} z}{1-e^{-2z}-2e^{-z}}\\
    &= 1+2e^{-z} - \frac{2(z + e^z -4)}{e^{z}(e^z-2)-1}.
  \end{align*}

\item $$\Pr(C_i \le z \mid  C_{i-1} \le z) = \frac{\Pr((X_{i} + X_{i+1})/2 \le z, (X_{i-1} + X_i)/2 \le z)}{\Pr((X_{i-1} + X_i)/2 \le z)}.$$
First, consider the denominator:
\begin{align*}
  \Pr((X_{i-1} + X_i)/2 \le z) &= \int_0^z 2e^{-2t}\left(\int_0^{z-t}2e^{-2s}\,ds\right)\,dt\\
  &= 1 - e^{-2 z} (1 + 2 z).
\end{align*}
Now, the numerator. First, we'll require $X_{i-1}/2 \le z$, then $X_{i}/2 \le z- X_{i-1}/2$, then $X_{i+1}/2 \le z-X_{i}/2$:
\begin{align*}
  \Pr((X_{i} + X_{i+1})/2 \le z, (X_{i-1} + X_i)/2 \le z) &= \int_{0}^z 2e^{-2t} \left(\int_{0}^{z-t} 2e^{-2s} \left[\int_{0}^{z-s} 2e^{-2r} \,dr\right]\, ds\right) \,dt\\
  &= 1 - e^{-4 z} - 4 e^{-2 z} z.
\end{align*}
Thus,
\begin{align*}
  \Pr(C_i \le z \mid  C_{i-1} \le z) &= \frac{1 - e^{-4 z} - 4 e^{-2 z} z}{1 - e^{-2 z} (1 + 2 z)}.
\end{align*}
\item \begin{align*}
 \Pr(R_n \le z \mid C_{n-2}\le z) &= \frac{\Pr(X_{n-1}/2 + X_n \le z, (X_{n-2}+ X_{n-1})/2 \le z) }{\Pr((X_{n-2}+ X_{n-1})/2 \le z)}
 \end{align*}
We already know the denominator from the previous step. We also know the numerator, by symmetry with the numerator in step 2. Thus:
\begin{align*}
  \Pr(R_n \le z \mid C_{n-2}\le z)&= \frac{1 - 2 e^{-3 z} + 3 e^{-2 z} - 2 e^{-z} - 2 e^{-2 z} z}{1 - e^{-2 z} (1 + 2 z)}.
\end{align*}
\end{enumerate}
Putting the four pieces together and simplifying: 
\begin{align*}
   \Pr(M_k \le z) = &\left(1-e^{-2z}-2e^{-z}\right)\left(\frac{1 - 2 e^{-3 z} + 3 e^{-2 z} - 2 e^{-z} - 2 e^{-2 z} z}{1 - e^{-2 z} (1 + 2 z)} 
 \right) \\
 &\qquad \cdot\left(1+2e^{-z} - \frac{2(z + e^z -4)}{e^{z}(e^z-2)-1}\right)\left(\frac{1 - e^{-4 z} - 4 e^{-2 z} z}{1 - e^{-2 z} (1 + 2 z)}\right)^{n-4}\\
 &= \frac{e^{-4 z} (-2 + e^z (3 + e^z (-2 + e^z) - 2 z))^2}{-1 + e^{2 z} - 2 z}\left(\frac{1 - e^{-4 z} - 4 e^{-2 z} z}{1 - e^{-2 z} (1 + 2 z)}\right)^{n-4}.\tag{$*$}
\end{align*}
We want to take the limit of $(*)$ as $n \rightarrow \infty$. We'll focus on the second part first, since the limit of a product is the product of the limits (as we will see, both limits are well-defined). Define 
\begin{align*}
\ell(z) &= \lim_{n\rightarrow \infty} \left(\frac{1 - e^{-4 z} - 4 e^{-2 z} z}{1 - e^{-2 z} (1 + 2 z)}\right)^{n-4}.
\end{align*}
Take the log to handle the exponent:
\begin{align*}
  \log \ell(z) &= \lim_{n\rightarrow \infty} (n-4)\log\left(\frac{1 - e^{-4 z} - 4 e^{-2 z} z}{1 - e^{-2 z} (1 + 2 z)}\right)\\
  &= \lim_{n\rightarrow \infty} [(n-4)\log\left(1 - e^{-4 z} - 4 e^{-2 z} z \right) - (n-4) \log\left(1 - e^{-2 z} (1 + 2 z)\right)].
\end{align*}

Now we'll split the limit into its two terms and plug in $(\log n + \log \log n+x) / 2$ for $z$ in $\ell(z)$. The first term:
\begin{align*}
  &\lim_{n\rightarrow \infty} (n-4)\log\left(1 - e^{-4 (\log n + \log \log n+ x) / 2 } - 4 e^{-2 (\log n + \log \log n+ x) / 2 } (\log n + \log \log n+ x) / 2  \right)\\
    &=\lim_{n\rightarrow \infty} (n-4)\log\left(1-e^{-2x}n^{-2}\log^{-2}n - 2e^{-x} n^{-1}\log^{-1}n (\log n + \log \log n + x)\right)\tag{simplify}\\
    &=\lim_{n\rightarrow \infty} \frac{\log\left(1-e^{-2x}n^{-2}\log^{-2}n - 2e^{-x} n^{-1}\log^{-1}n (\log n + \log \log n + x)\right)}{(n-4)^{-1}}\tag{rearrange}\\
  &=\lim_{n\rightarrow \infty} -(n-4)^2 \frac{d}{dn}\log\left(1-e^{-2x}n^{-2}\log^{-2}n - 2e^{-x} n^{-1}\log^{-1}n (\log n + \log \log n + x)\right)\tag{l'H\^{o}pital's rule}\\
  &= \lim_{n\rightarrow \infty} -(n-4)^2 \frac{2e^{x}n\log n (\log n + 1) (\log n + \log \log n + x -1) + 2\log n + 2}{ e^{x}n^2\log^2 n(e^{x}n\log n - 2\log n - 2 \log \log n -2x) - n \log n}\tag{take derivative}\\
    &= \lim_{n\rightarrow \infty} \frac{-2e^{x}n^3\log^3 n - O(n^3\log^2 n)}{e^{2x} n^3 \log^3n - O(n^2\log^3 n)}\tag{isolate highest order terms}\\
    &= -2e^{-x}.
\end{align*}
Plugging $(\log n + \log \log n+ x) / 2 $ into the second term and following the same strategy as above:
\begin{align*}
  &\lim_{n\rightarrow \infty}(n-4) \log\left(1 - e^{-2 ((\log n + \log \log n+ x) / 2)} (1 + 2 ((\log n + \log \log n+ x) / 2 ))\right)\\
  &= \lim_{n\rightarrow \infty} -(n-4)^2 \frac{d}{dn}\log\left(1 - e^{-x}n^{-1}\log^{-1}n (1 + \log n + \log \log n+ x)\right)\\
  &= \lim_{n\rightarrow \infty} (n-4)^2 \frac{(\log n + 1)(\log n + \log \log n + x)}{n \log n \left(-e^{x}n\log n + \log n + \log \log n + x + 1 \right)}\\
  &= \lim_{n\rightarrow \infty}  \frac{n^2\log^2 n + O(n^2 \log n \log \log n)}{-e^{x} n^2\log^2 n + O(n \log^2 n)}\\
  &= -e^{-x},
\end{align*}
where splitting the limit is allowed because both limits are finite. Thus, 
\begin{align*}
  \log \ell((\log n + \log \log n+ x) / 2 ) &= -2e^{-x} - (-e^{-x})\\
  &= -e^{-x}.
\end{align*}
We then have $\ell((\log n + \log \log n+ x) / 2 ) = e^{-e^{-x}}$. Going back to the first part of $(*)$ (recall that we are plugging in $z = (\log n + \log \log n + x) / 2$),
\begin{align*}
  \lim_{n \rightarrow \infty} \frac{e^{-4 z} (-2 + e^z (3 + e^z (-2 + e^z) - 2 z))^2}{-1 + e^{2 z} - 2 z} &= \lim_{z \rightarrow \infty} \frac{e^{-4 z} (-2 + e^z (3 + e^z (-2 + e^z) - 2 z))^2}{-1 + e^{2 z} - 2 z} \\
  &=  \lim_{z \rightarrow \infty} \frac{e^{2z} - O(e^z)}{e^{2 z} - O(z)}\\
  &= 1.
\end{align*}

Combining these findings gives us the limits of $(*)$ as $n\rightarrow \infty$, again plugging in $z = (\log n + \log \log n + x)/2$ and using the results above,
\begin{align*}
\lim_{n \rightarrow \infty} \Pr\left(M_k \le  \frac{\log n + \log \log n + x}{2}\right) = e^{-e^{-x}}.  
\end{align*}
As we saw at the beginning of the proof, to convert from the max sum of exponential RVs to the max plurality vote share (in the limit), we simply multiply by $n$. We can additionally convert back to $k+1 = n$ to prove the claim:
\begin{align*}
 e^{-e^{-x}} &= \lim_{n \rightarrow \infty} \Pr\left(M_k \le  \frac{\log n + \log \log n + x}{2}\right)\\
&=  \lim_{n \rightarrow \infty} \Pr\left(nV_k \le  \frac{\log n + \log \log n + x}{2}\right)  \\
&=  \lim_{k \rightarrow \infty} \Pr\left(V_k \le  \frac{\log (k+1) + \log \log (k+1) + x}{2(k+1)}\right).
\end{align*}
\end{proof}

Since $e^{-e^{-x}}\rightarrow 1$ and $e^{-e^{x}}\rightarrow 0$ as $x \rightarrow \infty$, we immediately have the following corollary of \Cref{lemma:gumbel}.

\begin{corollary}\label{cor:winner-vote-share}
For any function $g(k)$ with $\lim_{k \rightarrow \infty} g(k) = \infty$,
 \begin{equation*}
    \lim_{k\rightarrow \infty} \Pr\left( \frac{\log (k+1) + \log \log (k+1) - g(k)}{2(k+1)} \le V_k \le \frac{\log (k+1) + \log \log (k+1) + g(k)}{2(k+1)} \right) = 1.
 \end{equation*}
\end{corollary}

Intuitively, \Cref{cor:winner-vote-share} states that the asymptotic winning plurality vote share is almost exactly $\frac{\log(k+1) + \log \log (k+1)}{2(k+1)}$ with probability 1. We now provide a useful lemma before proving \Cref{thm:plurality-asymptotic-uniform}.

\begin{lemma}\label{lemma:pr-limit-1}
  Let $A_t$ and $B_t$ be events with $\lim_{t \rightarrow \infty} \Pr(A_t) = 1$ and $\lim_{t \rightarrow \infty} \Pr(B_t) = \beta > 0$. Then $\lim_{t \rightarrow \infty} \Pr(A_t \mid B_t) = 1.$
\end{lemma}
\begin{proof}
Using the Law of Total Probability and some basic probability facts,
\begin{align*}
  \Pr(A_t \mid B_t) &= \frac{\Pr (A_t \cap B_t)}{\Pr(B_t)}\\
  &= \frac{\Pr(B_t) - \Pr(\overline{A_t} \cap B_t)}{\Pr(B_t)}\\
  &\ge 1 - \frac{\Pr(\overline{A_t})}{\Pr(B_t)}
\end{align*}  
Thus,
\begin{align*}
  \lim_{t \rightarrow \infty} \Pr(A_t \mid B_t) &\ge \lim_{t \rightarrow \infty} 1- \frac{\Pr(\overline{A_t})}{\Pr(B_t)}\\
  &= 1 - \frac{0}{\beta}\\
  &= 1.
\end{align*}
We then have $\lim_{t \rightarrow \infty} \Pr(A_t \mid B_t) = 1$.

\end{proof}

Finally, we can prove \Cref{thm:plurality-asymptotic-uniform}.
\begin{proof}[Proof of \Cref{thm:plurality-asymptotic-uniform}]
 Consider a plurality election with $k$ candidates on the circle with circumference $1$, with points on the circle mapped to the interval $[0, 1)$ (we'll say the point on the circle corresponding to the endpoints of the interval maps to 0). Let $C_k$ be the position of the plurality winner on the circle with candidates positioned uniformly at random. By rotational symmetry, $C_k$ is uniform over the interval $[0, 1)$. Consider a particular configuration of $k$ candidates on the circle. When we break the circle to make it the unit interval, we only change the vote shares of the candidates closest to 0 and 1 (call them $x_\ell$ and $x_r$, respectively). Let $x_{\ell'}$ be the second-closest candidate to 0 and let $x_{r'}$ be the second-closest candidate to 1.
  
Consider $\Pr(x_{\ell'} < y) = 1 - \Pr(x_{\ell'} \ge y)$. The event $x_{\ell'} \ge y$ can be partitioned into two cases: either $x_\ell < y$ or $x_\ell \ge y$. Thus,
\begin{align*}
  \Pr(x_{\ell'} \ge y) &= \Pr(x_{\ell'} \ge y, x_\ell <y) + \Pr(x_{\ell'} \ge y, x_\ell \ge y)\\
  &= \Pr(x_\ell <y )\Pr(x_{\ell'} \ge y \mid x_\ell <y) + (1-y)^k\\
  &= (1 - (1-y)^k) (1-y)^{k-1}+ (1-y)^k.
\end{align*}
Now, pick $y = \frac{\log k}{4k}$. We then have:
\begin{align*}
  \lim_{k \rightarrow \infty}\Pr\left(x_{\ell'} < \frac{\log k}{4k}\right) &= \lim_{k \rightarrow \infty}\left[1 - \Pr\left(x_{\ell'} \ge \frac{\log k}{4k}\right)\right]\\
  &= \lim_{k \rightarrow \infty} \left[1 - \left(1 - \left(1-\frac{\log k}{4k}\right)^k\right) \left(1-\frac{\log k}{4k}\right)^{k-1} - \left(1-\frac{\log k}{4k}\right)^k \right]
\end{align*}
Note that $\lim_{k \rightarrow \infty}(1- \log k / (4k))^k = 0$, since $\lim_{k \rightarrow \infty}(1- x/k)^k = e^{-x}$. We also have $\lim_{k \rightarrow \infty}(1- \log k / (4k))^{k-1} = 0$, since $(1- \log k / (4k))^{k-1} = (1- \log k / (4k))^{k} / (1- \log k / (4k))$ and $\lim_{k \rightarrow \infty}(1- \log k / (4k)) = 1$. This means
\begin{align}
  \lim_{k \rightarrow \infty}\Pr\left(x_{\ell'} < \frac{\log k}{4k}\right) &= \lim_{k \rightarrow \infty} \left[1 - \left(1 - \left(1-\frac{\log k}{4k}\right)^k\right) \left(1-\frac{\log k}{4k}\right)^{k-1} - \left(1-\frac{\log k}{4k}\right)^k \right]\notag\\
  &= 1 - (1 - 0)\cdot 0 - 0\notag\\
  & = 1 \label{eq:left-share-small}.
\end{align}
Symmetrically, we also have $\lim_{k \rightarrow \infty}\Pr\left(x_{r'} > 1-\frac{\log k}{4k}\right) = 1$. We can therefore bound the asymptotic vote shares of $x_\ell$ and $x_r$ on both the circle and the unit interval. Neither can get more votes (in either setting) than the distance between $x_{\ell'}$ and $x_{r'}$ on the circle; i.e., $b = x_{\ell'} + 1- x_{r'}$ is an upper bound on the vote shares of $x_\ell$ and $x_r$ on both the circle and the unit interval. We can use the above facts to find an asymptotic bound on $b$. First, consider the probability that \emph{both} $x_{r'}$ and $x_{\ell'}$ are close to the boundaries:
\begin{align*}
  &\lim_{k \rightarrow \infty}\Pr\left(x_{r'} > 1-\frac{\log k}{4k}, x_{\ell'} < \frac{\log k}{4k}\right)\\
  &= \lim_{k \rightarrow \infty}\Pr\left(x_{\ell'} < \frac{\log k}{4k}\right) \cdot \lim_{k \rightarrow \infty}\Pr\left(x_{r'} > 1-\frac{\log k}{4k}\mid x_{\ell'} < \frac{\log k}{4k}\right)\\
  &= \lim_{k \rightarrow \infty}\Pr\left(x_{r'} > 1-\frac{\log k}{4k}\mid x_{\ell'} < \frac{\log k}{4k}\right) \tag{by \Cref{eq:left-share-small}}\\
  &= 1. \tag{by \Cref{lemma:pr-limit-1}}
\end{align*}
If $x_{r'} > 1-\frac{\log k}{4k}$, then $1 - x_{r'} < \frac{\log k}{4k}$. Thus, if both $x_{r'} > 1-\frac{\log k}{4k}$ and $ x_{\ell'} < \frac{\log k}{4k}$, we then have $b =  x_{\ell'} + 1- x_{r'} < \frac{\log k}{2k}$. Therefore $\lim_{k \rightarrow \infty }\Pr(b < \frac{\log k}{2k}) = 1$. That is, the asymptotic vote share of the leftmost and rightmost candidates are both less than $\frac{\log k}{2k}$ with probability 1. Meanwhile, we know from \Cref{lemma:gumbel} that the asymptotic winning vote share on the unit interval is larger than $\frac{\log k}{2k}$ with probability 1; i.e., with probability 1, neither $x_\ell$ nor $x_r$ is the winner (on either the circle or unit interval). Since no other vote shares change when we go between the unit interval and the circle, the winner on the unit interval is the same as the winner on the circle with probability 1 as $k \rightarrow \infty$. Thus, $\lim_{k \rightarrow \infty} \Pr(P_k \le x) = \lim_{k \rightarrow \infty} \Pr(C_k \le x) = x$.  
\end{proof}

\begin{proof}[Proof of \Cref{thm:moderate-voters-exclusion}]
  Begin the same way as in the proof of \Cref{thm:general-exclusion}, minimizing $x$'s vote shares with candidates at $c- \epsilon$ and $1-c + \epsilon$. We thus have
\begin{align*}
v(x) = F\left(\frac{x + 1-c + \epsilon}{2}\right) - F\left(\frac{c - \epsilon + x}{2}\right).
  \end{align*}
  
  If $f$ is monotonic and non-decreasing over $[0, 1/2]$, then $x$ has the smallest vote share when $x=c$. At this edge of the interval, $x$'s vote share is at least

\begin{align*}
    v(x) &\ge F\left(\frac{c + 1-c + \epsilon}{2}\right) - F\left(\frac{c - \epsilon + c}{2}\right) \\
    &= F\left(\frac{1+\epsilon}{2}\right) - F\left(c - \frac{\epsilon}{2}\right)\\
    &> F\left(1/2\right) - F\left(c\right)\tag{$F$ increasing}\\
    &= 1/2 - F\left(c\right). \tag{symmetry of $f$}
\end{align*}
Suppose $c\le F^{-1}(1/6)$. Then:
\begin{align*}
  1/2 - F\left(c\right) &\ge 1/2 - F(F^{-1}(1/6)) \\
  &= 1/2 - 1/6 \\
  &= 1/3.
\end{align*}
Thus $x$ cannot be eliminated next. The IRV winner must therefore be in $[c, 1-c]$ by the same argument as in \Cref{thm:1/6}.
\end{proof}

\begin{proof}[Proof of \Cref{thm:polarized-voters-exclusion}]
 As in \Cref{thm:general-exclusion}, minimize $x$'s vote shares with candidates at $c- \epsilon$ and $1-c + \epsilon$. If $f$ is monotonic and non-increasing over $[0, 1/2]$, then $x$ has the smallest vote share when $x=1/2$: 
  
  \begin{align*}
   v(x) &\ge F\left(\frac{x + 1-c+\epsilon}{2}\right) - F\left(\frac{c - \epsilon + x}{2}\right) \\
   &= F\left(\frac{3}{4} - \frac{c - \epsilon}{2}\right) - F\left(\frac{1}{4}+\frac{c - \epsilon}{2}\right)\\
  &= 2\left[F(1/2) - F\left(1/4 + \frac{c - \epsilon}{2}\right)\right]\tag{symmetry of $f$}\\
  &= 1 - 2F\left(1/4 + \frac{c - \epsilon}{2}\right).\tag{symmetry of $f$}
\end{align*}
Suppose $c\le 2 (F^{-1}(1/3) - 1/4)$. Then we have:
\begin{align*}
  1 - 2F\left(1/4 + \frac{c - \epsilon}{2}\right) &\ge 1 - 2F\left(1/4 + \frac{2 (F^{-1}(1/3) - 1/4) - \epsilon}{2}\right)\\
  &= 1 - 2F\left(F^{-1}(1/3) - \frac{\epsilon}{2}\right)\\
  &> 1 - 2F\left(F^{-1}(1/3)\right)\tag{$F$ increasing}\\
  &= 1 - 2/3\\
  &= 1/3.
\end{align*}
As before, $x$ cannot be eliminated next and some candidate in $[c, 1-c]$ must win under IRV. 
\end{proof}

\begin{proof}[Proof of \Cref{thm:very-polarized-voters-exclusion}]
 Suppose there is exactly one candidate $\ell \in [0, c]$ and at least one candidate each in $(c, 1-c)$ and $[c, 1]$ (if there are no candidates in $(c, 1-c)$, the claim is vacuously true). The smallest vote share $\ell$ could have occurs when $\ell = 0$ and there is a candidate at $c + \epsilon$. In this case, $\ell$'s vote share is
\begin{align*}
F\left(\frac{c+\epsilon}{2}\right) &> F\left(c/2\right).
\end{align*} 
If $c \ge 2F^{-1}(1/3)$, then
\begin{align*}
  F\left(c/2\right) &\ge F\left(2F^{-1}(1/3)/2\right)\\
  &= 1/3.
\end{align*}
Thus, $\ell$ cannot be eliminated next. By a symmetric argument, the last candidate $r$ in $[1-c, 1]$ is guaranteed more than a third of the vote. As long we we begin with at least one candidate in $[0, c]$ and at least one candidate in $[1-c, 1]$, once there is only one candidate remaining in each of these intervals, they will survive elimination until all candidates in $(c, 1-c)$ are eliminated.  At this point,  the IRV winner is guaranteed to be in $[0, c]$ or $[1-c, 1]$. 
\end{proof}

\begin{proof}[Proof of \Cref{thm:no-exclusion}]
 First, if $x_1=0$, consider a voter distribution with density function $f$ that increases monotonically over $[0, 1]$. Let $r$ be the position of the candidate immediately to the right of $x_1$. The candidate at $r$ must have a higher vote share than $x_1$, since they split the interval $[0, r]$ and the right half of this interval has more voter mass (as $f$ is increasing). Thus, $x_1$ cannot win under plurality, regardless of how many additional candidates we add. A symmetric argument shows that there are cases where a candidate at $x_1=1$ cannot win under plurality. 
 
 Now we show how to add candidates to the initial set $x_1, \dots, x_\kappa$ so that $x_1$ becomes the plurality winner if $x_1 \notin \{0, 1\}$. Since $f$ is continuous and $f(x_1) > 0$, there must exist some $\delta > 0$ such that $|f(x_1) - f(x_1+t)|  < f(x_1)/ 4$ for $t \in [-\delta, \delta]$. This argument still holds if we make $\delta$ smaller, so we ensure than $\delta < \min\{x_1, 1-x_1\}$ (both are positive since $x_1 \notin\{0, 1\}$). Now add candidates $\ell = x_1 - \delta$ and $r = x_1 + \delta$. Then the vote share $\ell$ gets on its right is less than $\frac{5}{8}\delta f(x_1)$:
   \begin{align*}
     \int_{\ell}^{(x_1 + \ell)/2}f(t) dt &< \int_{x_1-\delta}^{(x_1 + x_1-\delta)/2}(f(x_1) + f(x_1)/4) dt\\
     &= \frac{5}{4}f(x_1)(x_1-\delta/2 - x_1 + \delta)\\
     &= \frac{5}{8}\delta f(x_1).
   \end{align*}
   The same argument shows the vote share $r$ gets on its left is less than $\frac{5}{8}\delta f(x_1)$. Meanwhile, the vote share of $x_1$ is more than $\frac{3}{4}\delta f(x_1)$:
   \begin{align*}
     \int_{(x_1 + \ell)/2}^{(x_1 + r)/2} f(t) dt &> \int_{(x_1 + x_1 - \delta)/2}^{(x_1 + x_1+\delta)/2} (f(x_1)- f(x_1)/4) dt\\
     &= \frac{3}{4}f(x_1)(x_1 + \delta / 2 - x_1 + \delta / 2)\\
     &= \frac{3}{4} \delta f(x_1).
   \end{align*}
   Thus $x_1$ has a higher vote share than $\ell$ gets on its right and than $r$ gets on its left. Now, we repeatedly add candidates to the left of $\ell$ adjacent to to the candidates with the maximum vote shares in $[0, \ell]$. We can make the maximum vote share in $[0, \ell]$ arbitrarily small (except $\ell$'s) by adding enough candidates in this way---in particular, we can make it smaller than $v(x_1)$. We can also make $\ell$'s total vote share smaller than $v(x_1)$, since the vote share $\ell$ gets on its right is strictly smaller than $v(x_1)$. Doing the same in $[r, 1]$ then ensures $x_1$ is the plurality winner.
 \end{proof}

\section{Derivations of $f_{P_3}$ and $f_{R_3}$}\label{app:k3_derivations}
In this section, we derive the probability density functions of the position of the plurality and IRV winners for  1-Euclidean profiles with $k=3$ uniformly placed candidates and continuous uniform voters. We then compute the variances of these distributions. The calculations for plurality are in \Cref{sec:k-3-plurality-dsn} and the calculations for IRV are in \Cref{sec:k-3-irv-dsn}. First, we provide an overview of our approach. 

Let $X_1, \dots, X_k \sim \text{Unif}(0, 1)$ be the random positions of the $k$ candidates and let $W \in \{X_1, \dots, X_k\}$ be the position of the winner. Let $X_{(i)}$ denote the $i$th order statistic of $X_1, \dots, X_k$.

The density of the winner's position at a point $w$, denoted $f(w)$, is $k$ times the probability that a particular candidate at $w$ is the winner (times the density of that candidate's position at $w$, which is 1). We sum over the possible order statistics of the winner to compute $f(w)$:

\begin{align*}
  f(w) &= k \Pr(w \text{ wins} )\\
  &= k\sum_{i = 1}^k \Pr(w \text{ wins}, w=X_{(i)})
\end{align*}

What is the probability that a candidate at position $w$ with order statistic $i$ wins? Say the winner is candidate 1. We can choose which $i - 1$ candidates are to their left. The remaining $k - i$ candidates are to their right. Then, we integrate over the positions of the other candidates where the candidate at $w$ wins.

\begin{align*}
  &\Pr(w \text{ wins}, w = X_{(i)}) \\&= \binom{k}{i - 1} \underbrace{\int_{0}^{x_i}\dots \int_{0}^{w}}_{i-1}\underbrace{\int_{x_i}^1\dots \int_{w}^1}_{k - 1} \mathbf{1}[w\text{ wins given positions }x_2, \dots, x_k] dx_{k}\dots dx_{2}
\end{align*}

We can also note that the win probability (and therefore the winner density) is symmetric about 0.5: $f(w) = 1-f(w)$. We therefore only need to consider $w\in [0, 0.5]$.

\subsection{1d plurality winner distribution, $k=3$}\label{sec:k-3-plurality-dsn}
We'll compute $\Pr(w \text{ wins}, w=X_{(i)})$ for $i = 1, 2, 3$ and $w \in [0, 0.5]$. That is, we'll compute the win probability of a candidate at a point $w$ in the cases where they are the leftmost, the middle, and the rightmost of the three candidates. In all cases, we'll call the winner candidate $1$ and the losers candidates 2 and 3, at positions $x_2$ and $x_3$.

\begin{enumerate}
  \item $w=X_{(1)}$. Consider the order $w < x_2 < x_3$ (we'll multiply by 2 later to account for the ordering $w < x_3 < x_2$). For $w$ to beat $x_2$, we need:
\begin{align*}
 & w + (x_2 - w) / 2 > (x_2 - w) / 2 + (x_3 - x_2) / 2\\
 \Leftrightarrow  \quad & w > (x_3 - x_2) / 2\\
 \Leftrightarrow  \quad & 2w > x_3 - x_2\\
 \Leftrightarrow  \quad & x_3 <  2w + x_2 
 \stepcounter{equation}\tag{\theequation}\label{eq:case1_beatx2}
 \end{align*}

For $w$ to beat $x_3$, we need 
\begin{align*}
  &w + (x_2 - w) / 2 > 1-x_3 + (x_3 - x_2) / 2\\
  \Leftrightarrow  \quad & 2w + x_2 - w > 2 - 2x_3 + x_3 - x_2\\
  \Leftrightarrow  \quad & w + x_2 > 2 - x_3 - x_2\\
    \Leftrightarrow  \quad & x_3 > 2  - 2x_2 - w
    \stepcounter{equation}\tag{\theequation}\label{eq:case1_beatx3}
  \end{align*}
  
  For both (\ref{eq:case1_beatx2}) and (\ref{eq:case1_beatx3}) to be feasible, $w$ and $x_2$ cannot both be too small. The inequalities match at 
\begin{align*}
  &2w +x_2 = 2-2x_2-w\\
  \Leftrightarrow  \quad & 3w + 3x_2 = 2\\
   \Leftrightarrow  \quad & w + x_2 = 2/3. 
\end{align*}
We therefore need $w+x_2 > 2/3$ to satisfy both (\ref{eq:case1_beatx2}) and (\ref{eq:case1_beatx3}). We summarize the constraints on $w$ and $x_2$ in the following plot, where the gray region contains points where $w$ can win (given $w < x_2 < x_3)$.

\begin{center}
  \begin{tikzpicture}
\begin{axis}[xmin=0, xmax=0.5,ymin=0, ymax=1, samples=10,xlabel=$w$, ylabel=$x_2$, legend pos=outer north east, legend cell align={left}]

  \addplot[red, thick, name path=b] (x,x);
  \addplot[blue, thick, name path=a] (x,2/3 -x);

   \path[name path=axis] (axis cs:0, 1) -- (axis cs:1, 1);

%  \addplot[gray!20] fill between[of=a and axis];
    \addplot[gray!20] fill between[of=b and axis];
    \addplot[white] fill between[of=b and a];

  \legend{$x_2 = w$, $x_2 = 2/3-w$}
\end{axis}
\end{tikzpicture}
\end{center}

The lower bound on $x_3$ for $w$ to win is the minimum of $x_2$ and $2 - 2x_2 - w$, while the upper bound is the maximum of $1$ and $2w+x_2$. We'll plot the lines where these bounds change, namely $x_2 = 2-2x_2 - w \Leftrightarrow x_2 = 2/3 - w/3$ and $1 = 2w+x_2 \Leftrightarrow x_2 = 1 - 2w $, and label the regions over which we can easily integrate:

\begin{center}
  \begin{tikzpicture}
\begin{axis}[xmin=0, xmax=0.5,ymin=0, ymax=1, samples=10,xlabel=$w$, ylabel=$x_2$, legend pos=outer north east, legend cell align={left}]

  \addplot[red, thick, name path=b] (x,x);
  \addplot[blue, thick, name path=a] (x,2/3 -x);
  \addplot[green, thick, name path=c] (x,1-2*x);
  \addplot[violet, thick, name path=c] (x,2/3 - x/3);

   \path[name path=axis] (axis cs:0, 1) -- (axis cs:1, 1);
   
   \draw[thick, dotted] (axis cs:0.2, 0) -- (axis cs:0.2, 1);
   \draw[thick, dotted] (axis cs:1/3, 0) -- (axis cs:1/3, 2/3- 1/9);

%  \addplot[gray!20] fill between[of=a and axis];
    \addplot[gray!20] fill between[of=b and axis];
    \addplot[white] fill between[of=b and a];
    
  \node at (axis cs:0.06, 0.75) {A};
  \node at (axis cs:0.14, 0.87) {B};
  \node at (axis cs:0.35, 0.78) {C};
  \node at (axis cs:0.16, 0.56) {D};
  \node at (axis cs:0.22, 0.49) {E};
  \node at (axis cs:0.29, 0.5) {F};
  \node at (axis cs:0.37, 0.47) {G};

  \legend{(a) $x_2 = w$, (b) $x_2 = 2/3-w$, (c) $x_2 = 1-2w$, (d) $x_2 = 2/3 - w/3$}
\end{axis}
\end{tikzpicture}
\end{center}

Above line $(c)$, the upper bound for $x_3$ is $1$; below $(c)$, it's $2w+x_2$. Above line $(d)$, the lower bound for $x_3$ is $x_2$; below, it's $2-2x_2-w$. Lines (c) and (d) intersect at $w = 1/5$, while lines (a), (b), and (c) intersect at $w=1/3$.

With this information in hand, we can compute the integral describing the win probability of $w$ by summing the integrals for regions A--G and multiplying by 2 to account for the ordering $w < x_3 < x_2$:

\begin{align*}
  A: \quad & \int_{2/3-w/3}^{1-2w} \int_{x_2}^{2w+x_2} \, dx_3 \,dx_2 &&= 2w/3 - 10w^2/3 \\
  B: \quad &  \int_{1-2w}^{1} \int_{x_2}^{1} \, dx_3 \,dx_2 &&=2w^2\\
  C: \quad &  \int_{2/3-w/3}^{1} \int_{x_2}^{1} \, dx_3 \,dx_2 &&= 1/18 + w/9 + w^2/18 \\
  D: \quad &  \int_{2/3-w}^{2/3-w/3} \int_{2-2x_2-w}^{2w+x_2} \, dx_3 \,dx_2 &&= 2 w^2/3\\
  E: \quad &  \int_{2/3-w}^{1-2w} \int_{2-2x_2-w}^{2w+x_2} \, dx_3 \,dx_2 &&= 1/6 - w + 3 w^2/2\\
  F: \quad & \int_{1-2w}^{2/3-w/3} \int_{2-2x_2-w}^{1} \, dx_3 \,dx_2 &&=-2/9 + 14 w/9 - 20 w^2/9\\
  G: \quad &  \int_{w}^{2/3-w/3} \int_{2-2x_2-w}^{1} \, dx_3 \,dx_2 &&= -2/9 + 14 w/9 - 20 w^2/9
\end{align*}

For $w \in [0, 1/5]$, the win probability is $2(2w/3 - 10w^2/3 + 2w^2 + 2 w^2/3) = 4 w/3 - 4 w^2/3$.

For $w \in [1/5, 1/3]$, the win probability is $2(1/6 - w + 3 w^2/2 + -2/9 + 14 w/9 - 20 w^2/9 + 1/18 + w/9 + w^2/18) = 4 w/3 - 4 w^2/3$.

Finally, for $w \in [1/3, 1/2]$, the win probability is $2(-2/9 + 14 w/9 - 20 w^2/9 + 1/18 + w/9 + w^2/18) = -1/3 + 10 w/3 - 13 w^2/3$.

To summarize, the win probability is:
\begin{equation}
  \Pr(w \text{ wins}, w = X_{(1)}) = \begin{cases}
  4 w/3 - 4 w^2/3, & w \in [0, 1/3]\\
  -1/3 + 10 w/3 - 13 w^2/3, & w \in [1/3, 1/2]
 \end{cases}
\end{equation} 

Visualizing this:

\begin{center}
  \begin{tikzpicture}
\begin{axis}[xmin=0, xmax=0.5,ymin=0, ymax=0.4, samples=100,xlabel=$w$, ylabel={$\Pr(w \text{ wins}, w=X_{(1)})$}, legend pos=outer north east, legend cell align={left},
 y tick label style={
        /pgf/number format/.cd,
            fixed,
            precision=1,
        /tikz/.cd
    },]

  \addplot[red, thick, name path=b][domain=0:1/3] (x,4*x/3 - 4 *x^2 / 3);
  
    \addplot[blue, thick, name path=b][domain=1/3:0.5] (x,-1/3 + 10*x / 3 - 13 * x^2 / 3);

  \legend{$4w/3-4w^2/3$, $ -1/3 + 10 w/3 - 13 w^2/3$}
\end{axis}
\end{tikzpicture}
\end{center}

\item $w=X_{(2)}$. Consider the ordering $x_2 < w < x_3$ (we'll multiply by 2 later to account for $x_3 < w < x_2$). For $w$ to beat $x_2$, we need:
\begin{align*}
  &(w - x_2) / 2 + (x_3 - w) / 2 > x_2 + (w-x_2)/2\\
  \Leftrightarrow \quad & (x_3 - w) / 2 > x_2\\
  \Leftrightarrow \quad & x_3 > 2x_2 + w
  \stepcounter{equation}\tag{\theequation}\label{eq:case2_beatx1}
\end{align*} 

In order for this to be feasible, we need $2x_2 + w < 1 \Leftrightarrow x_2 < (1-w) / 2$. For $w$ to beat $x_3$, we need:
\begin{align*}
  &(w - x_2) / 2 + (x_3 - w) / 2 > 1-x_3 + (x_3-w)/2\\
  \Leftrightarrow \quad & (w - x_2) / 2 > 1-x_3\\
  \Leftrightarrow \quad & x_3 > 1-(w-x_2) / 2
  \stepcounter{equation}\tag{\theequation}\label{eq:case2_beatx2}
\end{align*} 

These bounds are equal if
\begin{align*}
  &2x_2 + w = 1-(w-x_2) / 2\\
  \Leftrightarrow \quad & 4x_2 + 2w = 2 - w + x_2\\
  \Leftrightarrow \quad & 3x_2 + 3w = 2\\
  \Leftrightarrow \quad & x_2 = 2 / 3 - w
\end{align*}

Again, we can split the $w$--$x_2$ plane using this line to make integration easy. Note that the lines $x_2 = 2/3 - w$ and the line $x_2 = w$ intersect at $1/3$. 

\begin{center}
  \begin{tikzpicture}
\begin{axis}[xmin=0, xmax=0.5,ymin=0, ymax=0.5, samples=10,xlabel=$w$, ylabel=$x_2$, legend pos=outer north east, legend cell align={left}]

  \addplot[red, thick, name path=b] (x,x);
  \addplot[blue, thick, name path=a] (x,1/2-x/2);
    \addplot[green, thick, name path=c] (x,2/3-x);
    \draw[thick, dotted] (axis cs:1/3, 0) -- (axis cs:1/3, 1/3);

   \path[name path=axis] (axis cs:0, 0) -- (axis cs:1, 0);

    \addplot[gray!20] fill between[of=b and axis];
        \addplot[white] fill between[of=b and a];
        
      \node at (axis cs:0.24, 0.1) {A};
  \node at (axis cs:0.42, 0.1) {B};
  \node at (axis cs:0.46, 0.24) {C};
  \legend{(a) $x_2 = w$, (b) $x_2 = (1-w)/2$, (c) $x_2 = 2/3-w$}
\end{axis}
\end{tikzpicture}
\end{center}

Above line (c), the constraint $x_3 > 2x_2 + w$ dominates. Below line (c), $x_3 > 1-(w-x_2)/2$ dominates. Integrating in the regions A--C:

\begin{align*}
  A: \quad & \int_{0}^{w} \int_{1-(w-x_2)/2}^{1} \, dx_3 \,dx_2 &&= w^2/4 \\
  B: \quad &  \int_{0}^{2/3-w} \int_{1-(w-x_2)/2}^{1} \, dx_3 \,dx_2 &&= -1/9 + 2 w/3 - 3 w^2/4\\
  C: \quad &  \int_{2/3-w}^{(1-w)/2} \int_{2x_2+w}^{1} \, dx_3 \,dx_2 &&= 1/36 - w/6 + w^2/4 \\
\end{align*}

Recall that we need to multiply by 2 to account for the ordering $x_3 < w < x_2$.

For $w \in [0, 1/3]$, the win probability is $2(w^2/4) = w^2/2$.

For $w \in [1/3, 1/2]$, the win probability is $2(-1/9 + 2 w/3 - 3 w^2/4 + 1/36 - w/6 + w^2/4) = -1/6 + w - w^2$.

To summarize:

\begin{equation}
  \Pr(w \text{ wins}, w = X_{(2)}) = \begin{cases}
  w^2/2, & w \in [0, 1/3]\\
  -1/6 + w - w^2, & w \in [1/3, 1/2]
 \end{cases}
\end{equation} 

Visualizing this:
\begin{center}
  \begin{tikzpicture}
\begin{axis}[xmin=0, xmax=0.5,ymin=0, ymax=0.1, samples=100,xlabel=$w$, ylabel={$\Pr(w \text{ wins}, w=X_{(2)})$}, legend pos=outer north east, legend cell align={left},
 y tick label style={
        /pgf/number format/.cd,
            fixed,
            precision=2,
        /tikz/.cd
    },]

  \addplot[red, thick, name path=b][domain=0:1/3] (x,x^2 / 2);
  
    \addplot[blue, thick, name path=b][domain=1/3:0.5] (x,-1/6 + x - x^2);

  \legend{$w^2/2$, $ -1/6 + w - w^2$}
\end{axis}
\end{tikzpicture}
\end{center}

\item  $w=X_{(3)}$. This means $x_2 < w$ and $x_3 < w$. Since $w \le 0.5$, $w$ always wins. We also have $\Pr(x_2 < w) = w$ and $\Pr(x_3 < w)=w$ . Thus, $\Pr(w \text{ wins}, w = X_{(3)} \mid w \in [0, 0.5]) = w^2$.
\end{enumerate}

Adding the three cases, we arrive at $\Pr(w \text{ wins})$. For $w \in [0, 1/3]$, the sum is $4w/3 - 4w^2/3 + w^2/2 + w^2 = 4 w/3 + w^2/6$. For $w \in [1/3, 1/2]$, the sum is $-1/3 + 10w/3 - 13w^2/3 -1/6 + w- w^2 + w^2 = -1/2 + 13 w/3 - 13 w^2/3 $. Summarizing and plotting:

\begin{equation}
  \Pr(w \text{ wins}) = \begin{cases}
  4 w/3 + w^2/6, & w \in [0, 1/3]\\
  -1/2 + 13 w/3 - 13 w^2/3, & w \in [1/3, 1/2]
 \end{cases}
\end{equation} 

\begin{center}
  \begin{tikzpicture}
\begin{axis}[xmin=0, xmax=0.5,ymin=0, ymax=0.6, samples=100,xlabel=$w$, ylabel={$\Pr(w \text{ wins})$}, legend pos=outer north east, legend cell align={left},
 y tick label style={
        /pgf/number format/.cd,
            fixed,
            precision=1,
        /tikz/.cd
    },]

  \addplot[red, thick, name path=b][domain=0:1/3] (x,4*x/3 + x^2/6);
  
    \addplot[blue, thick, name path=b][domain=1/3:0.5] (x,-1/2 + 13 *x/3 - 13 *x^2/3);

  \legend{$4 w/3 + w^2/6$, $-1/2 + 13 w/3 - 13 w^2/3$}
\end{axis}
\end{tikzpicture}
\end{center}

Scaling by 3, the variance of $P_3$ is then:
\begin{equation}
  6\Bigg[\int_{0}^{1/3}(w - 1/2)^2(4w/3 + w^2/6) \,dw + \int_{1/3}^{1/2}(w - 1/2)^2(-1/2 + 13w/3 - 13w^2/3) \,dw \Bigg] = 23/540 \approx 0.043 
\end{equation}

\subsection{1d IRV winner distribution, $k=3$}\label{sec:k-3-irv-dsn}
We'll perform the same type of analysis, but for IRV instead of plurality. In addition to breaking down cases by the order statistic of the winner, we'll also consider the IRV elimination order. 

\begin{enumerate}
  \item $w=X_{(1)}$. Consider the order $w < x_2 < x_3$ (we'll multiply by 2 later to account for the ordering $w < x_3 < x_2$).
  \begin{enumerate}
  \item Candidate 2 is eliminated first. Since 2 has a smaller vote share than the winner,
 \begin{align*}
    &(x_2 - w) / 2 + (x_3 - x_2) / 2 < w + (x_2 - w) / 2 \\
    \Leftrightarrow \quad & (x_3 - x_2) / 2 < w\\
    \Leftrightarrow \quad & x_3 < 2w + x_2
    \stepcounter{equation}\tag{\theequation}\label{eq:irv_1a_loseto1}
  \end{align*}
  
  Since 2 has a smaller vote share than 3,
  \begin{align*}
    &(x_2 - w) / 2 + (x_3 - x_2) / 2 < 1-x_3+ (x_3 - x_2) / 2 \\
    \Leftrightarrow \quad & (x_2 - w) / 2 < 1-x_3\\
    \Leftrightarrow \quad & x_3 < 1 - (x_2 -w)/2. 
    \stepcounter{equation}\tag{\theequation}\label{eq:irv_1a_loseto3}
  \end{align*}
  For this constraint to be feasible, we need
  \begin{align*}
    & x_2 < 1 - (x_2 - w) / 2  \\
    \Leftrightarrow \quad & 2x_2 < 2 - x_2 + w \\
    \Leftrightarrow \quad & x_2 < 2/3 + w/3
  \end{align*}
  
  The constraints are equal if
  \begin{align*}
    &2w + x_2 = 1-(x_2 - w)/2\\
    \Leftrightarrow \quad & 4w + 2x_2 = 2 - x_2 + w\\
    \Leftrightarrow \quad & 3w + 3x_2 = 2\\
    \Leftrightarrow \quad & x_2 = 2/3 - w
  \end{align*}

  Once 2 is eliminated, whoever is closer to $1/2$ is the winner. Thus, for $w$ to win, we need $x_3 > 1-w$. This constraint equals constraint \ref{eq:irv_1a_loseto1} if
  \begin{align*}
    & 2w + x_2 = 1-w\\ 
    \Leftrightarrow \quad & x_2 = 1-3w
  \end{align*}
  If $w$ is too small (i.e., left of the line $x_2 = 1-3w$), then we cannot satisfy both $x_3 > 1-w$ and $x_3 < 2w + x_2$. The constraint $x_3 > 1-w$ equals constraint \ref{eq:irv_1a_loseto3} if
    \begin{align*}
    & 1 - (x_2 -w)/2 = 1-w\\ 
    \Leftrightarrow \quad & 2- x_2 + w = 2-2w\\
    \Leftrightarrow \quad & x_2 = 3w
  \end{align*}
  Again, if we are to the left of this line, we cannot satisfy both $x_3 > 1-w$ and  $x_3 < 1 - (x_2 -w)/2$. Finally, the lower bound on $x_3$ is the maximum of $1-w$ and $x_2$. Above the line $x_2 = 1-w$, the lower bound $x_2$ dominates; below, $1-w$ dominates.

  \begin{center}
  \begin{tikzpicture}
\begin{axis}[xmin=0, xmax=0.5,ymin=0, ymax=1, samples=10,xlabel=$w$, ylabel=$x_2$, legend pos=outer north east, legend cell align={left}]

  \addplot[red, thick, name path=a] (x,x);
  \addplot[blue, thick, name path=b] (x,2/3 +x/3);
  \addplot[green, thick, name path=c] (x,2/3 -x);
  \addplot[violet, thick, name path=d] (x,1 -3*x);
  \addplot[orange, thick, name path=e] (x,3*x);
  \addplot[magenta, thick, name path=f] (x,1-x);
%  \addplot[cyan, thick, name path=g] (x,1-2*x);

%  \draw[thick, dotted] (axis cs:2/9, 1/3) -- (axis cs:2/9, 4/9);
  \draw[thick, dotted] (axis cs:1/4, 1/4) -- (axis cs:1/4, 3/4);
  \draw[thick, dotted] (axis cs:1/3, 1/3) -- (axis cs:1/3, 2/3);
%  \draw[thick, dotted] (axis cs:2/5, 3/5) -- (axis cs:2/5, 4/5);

   \path[name path=top] (axis cs:0, 1) -- (axis cs:1, 1);
   \path[name path=bottom] (axis cs:0, 0) -- (axis cs:1, 0);

  \node at (axis cs:0.23, 0.385) {A};
%    \draw (axis cs:0.205,0.358) -- (axis cs:0.212, 0.395);

  \node at (axis cs:0.22, 0.54) {C};
    \node at (axis cs:0.275, 0.33) {B};

  \node at (axis cs:0.29, 0.53) {D};
  \node at (axis cs:0.38, 0.51) {E};
  \node at (axis cs:0.42, 0.69) {F};
%  \node at (axis cs:0.45, 0.68) {G};

%  \addplot[gray!20] fill between[of=a and axis];
    \addplot[gray!20] fill between[of=b and a];
    \addplot[white] fill between[of=bottom and d];
    \addplot[white] fill between[of=e and top];

  \legend{(a) $x_2 = w$, (b) $x_2 = 2/3+w/3$, (c) $x_2=2/3-w$, (d) $x_2=1-3w$, (e) $x_2 =3w$, (f) $x_2 = 1-w$}
\end{axis}
\end{tikzpicture}
\end{center}
Below line (c), constraint (\ref{eq:irv_1a_loseto1}) dominates; above line (c), constraint (\ref{eq:irv_1a_loseto3}) dominates. Below line (f), the lower bound on $x_3$ is $1-w$; above line (f), the lower bound is $x_2$. Lines (c), (d), (e) intersect at $w = 1/6$; lines (b), (e), (f) intersect at $w = 1/4$; lines (a) and (d) intersect at $w=1/4$; lines (a) and (c) intersect at $w=1/3$.

Integrating for each region and multiplying by 2 to account for the order $w < x_3 < x_2$:
\begin{align*}
  A: \quad & \int_{1-3w}^{2/3-w} \int_{1-w}^{2w+x_2} \, dx_3 \,dx_2 &&= 1/18 - 2 w/3 + 2 w^2 \\
  B: \quad &  \int_{w}^{2/3-w} \int_{1-w}^{2w+x_2} \, dx_3 \,dx_2 &&= -4/9 + 10 w/3 - 6 w^2 \\
  C: \quad &  \int_{2/3-w}^{3w} \int_{1-w}^{1 - (x_2 -w)/2} \, dx_3 \,dx_2 &&= 1/9 - 4 w/3 + 4 w^2\\
  D: \quad &  \int_{2/3-w}^{1-w} \int_{1-w}^{1 - (x_2 -w)/2} \, dx_3 \,dx_2 &&= -5/36 + 2 w/3\\
  E: \quad &  \int_{w}^{1-w} \int_{1-w}^{1 - (x_2 -w)/2} \, dx_3 \,dx_2 &&=  -1/4 + 2 w - 3 w^2 \\
  F: \quad & \int_{1-w}^{2/3+w/3} \int_{x_2}^{1 - (x_2 -w)/2} \, dx_3 \,dx_2 &&= 1/12 - 2 w/3 + 4 w^2/3
  \end{align*}

For $w\in [1/6, 1/4]$, the win probability is $2(1/18 - 2 w/3 + 2 w^2 + 1/9 - 4 w/3 + 4 w^2) = 1/3 - 4 w + 12 w^2$.

For $w \in [1/4, 1/3]$, the win probability is $2(-4/9 + 10 w/3 - 6 w^2 -5/36 + 2 w/3 + 1/12 - 2 w/3 + 4 w^2/3) = -1 + 20 w/3 - 28 w^2/3$.

For $w \in [1/3, 1/2]$, the win probability is $2(-1/4 + 2 w - 3 w^2 + 1/12 - 2 w/3 + 4 w^2/3) = -1/3 + 8 w/3 - 10 w^2/3$.

Summarizing and visualizing:

\begin{equation}
  \Pr(w \text{ wins}, w = X_{(1)}, X_{(2)} \text{ elim 1st}) = \begin{cases}
    1/3 - 4 w + 12 w^2, & w\in [1/6, 1/4]\\
    -1 + 20 w/3 - 28 w^2/3, & w \in [1/4, 1/3]\\
    -1/3 + 8 w/3 - 10 w^2/3, & w \in [1/3, 1/2]
  \end{cases}
\end{equation}

\begin{center}
  \begin{tikzpicture}
\begin{axis}[xmin=0, xmax=0.5,ymin=0, ymax=0.25, samples=100,xlabel=$w$, ylabel={$\Pr(w \text{ wins}, w = X_{(1)}, X_{(2)} \text{ elim 1st})$}, legend pos=outer north east, legend cell align={left},
 y tick label style={
        /pgf/number format/.cd,
            fixed,
            precision=2,
        /tikz/.cd
    },]

  \addplot[red, thick, name path=b][domain=1/6:1/4] (x,1/3 - 4 *x + 12 *x^2);
  \addplot[blue, thick, name path=b][domain=1/4:1/3] (x, -1 + 20 *x/3 - 28 *x^2/3);
  \addplot[green, thick, name path=b][domain=1/3:1/2] (x, -1/3 + 8 *x/3 - 10 *x^2/3);

  \legend{$1/3 - 4 w + 12 w^2$, $-1 + 20 w/3 - 28 w^2/3$ , $-1/3 + 8 w/3 - 10 w^2/3$}
\end{axis}
\end{tikzpicture}
\end{center}

  \item Candidate 3 is eliminated first. Since candidate 3 has a smaller vote share than candidate 2:
  \begin{align*}
    &(x_3 - x_2) / 2 + 1-x_3 < (x_3 - x_2) / 2 + (x_2 - w) / 2\\
    \Leftrightarrow \quad &  1-x_3 <  (x_2 - w) / 2\\
    \Leftrightarrow \quad &  x_3 > 1-  (x_2 - w) / 2
  \end{align*} 
  Since candidate 3 has a smaller vote share than the winner:
    \begin{align*}
    &(x_3 - x_2) / 2 + 1-x_3 < w + (x_2 - w) / 2\\
    \Leftrightarrow \quad & x_3 - x_2 + 2-2x_3 < 2w + x_2 - w \\
    \Leftrightarrow \quad &  -x_3+ 2 < w + 2x_2\\
    \Leftrightarrow \quad &  x_3 > 2 - w - 2x_2
  \end{align*} 
  For this to be feasible, we need:
  \begin{align*}
    & 2 - w - 2x_2 < 1\\
    \Leftrightarrow \quad & 1 - w < 2x_2\\
    \Leftrightarrow \quad & x_2 > (1 - w) / 2
  \end{align*}
  
  The constraints $ x_3 > 1-  (x_2 - w) / 2$ and $x_3 > 2 - w - 2x_2$ are equal when:
  \begin{align*}
    & 1-  (x_2 - w) / 2 = 2 - w - 2x_2\\
    \Leftrightarrow \quad & 2-  x_2 + w = 4 - 2w - 4x_2\\
    \Leftrightarrow \quad & 3x_2  = 2 - 3w\\
     \Leftrightarrow \quad & x_2  = 2/3 - w
  \end{align*}
  Above the line $x_2  = 2/3 - w$, the constraint $x_3 > 1-  (x_2 - w) / 2$ dominates; below, $x_3 > 2 - w - 2x_2$ dominates. 

In order for $w$ to win, it must be closer to the center than $x_2$. This requires that $x_2 > 1-w$ (since $w < 0.5$ and $w < x_2$). Thus, we never need to worry about the constraint $x_3 < 2-w-2x_2$ (since the line $x_2 = 1-w$ is above the line $2/3 -w$). The lower bound on $x_3$ is thus the maximum of $x_2$ and $1-  (x_2 - w) / 2$. These are equal if
\begin{align*}
  & x_2 = 1-  (x_2 - w) / 2\\
  \Leftrightarrow \quad &  2x_2 = 2-  x_2 + w\\
  \Leftrightarrow \quad & 3x_2 = 2 + w\\
  \Leftrightarrow \quad & x_2 = 2/3 + w/3
\end{align*}
Above the line $x_2 = 2/3 + w/3$, the lower bound on $x_3$ is $x_2$; below, it's $1-  (x_2 - w) / 2$.

  \begin{center}
  \begin{tikzpicture}
\begin{axis}[xmin=0, xmax=0.5,ymin=0, ymax=1, samples=10,xlabel=$w$, ylabel=$x_2$, legend pos=outer north east, legend cell align={left}]

  \addplot[red, thick, name path=a] (x,x);
  \addplot[blue, thick, name path=b] (x,1/2 - x/2);
  \addplot[green, thick, name path=c] (x,2/3 -x);
  \addplot[violet, thick, name path=d] (x,1-x);
  \addplot[orange, thick, name path=e] (x,2/3 + x/3);

  \draw[thick, dotted] (axis cs:1/4, 3/4) -- (axis cs:1/4, 1);
%  \draw[thick, dotted] (axis cs:1/4, 1/4) -- (axis cs:1/4, 5/12);
%  \draw[thick, dotted] (axis cs:1/3, 1/3) -- (axis cs:1/3, 2/3);
%  \draw[thick, dotted] (axis cs:2/5, 3/5) -- (axis cs:2/5, 4/5);

   \path[name path=top] (axis cs:0, 1) -- (axis cs:1, 1);
   \path[name path=bottom] (axis cs:0, 0) -- (axis cs:1, 0);

  \node at (axis cs:0.19, 0.9) {A};
%    \draw (axis cs:0.205,0.358) -- (axis cs:0.212, 0.395);
%
  \node at (axis cs:0.35, 0.9) {B};
  \node at (axis cs:0.43, 0.7) {C};
%  \node at (axis cs:0.29, 0.46) {D};
%  \node at (axis cs:0.38, 0.51) {E};
%  \node at (axis cs:0.38, 0.68) {F};
%  \node at (axis cs:0.45, 0.68) {G};
%

  \addplot[gray!20] fill between[of=d and top];
%    \addplot[gray!20] fill between[of=b and a];
%    \addplot[white] fill between[of=bottom and d];
%    \addplot[white] fill between[of=d and e];

  \legend{(a) $x_2 = w$, (b) $x_2 = (1-w)/2$, (c) $x_2 = 2/3-w$, (d) $x_2 = 1-w$, (e) $x_2 = 2/3 + w/3$}
\end{axis}
\end{tikzpicture}
\end{center}
Lines (d) and (e) intersect at $w = 1/4$. Integrating over the regions:
\begin{align}
  A: &\int_{1-w}^{1} \int_{x_2}^{1} \, dx_3 \,dx_2 &&= w^2/2 \\
  B: &\int_{2/3+w/3}^{1} \int_{x_2}^{1} \, dx_3 \,dx_2 &&= 1/18 - w/9 + w^2/18\\
  C: &\int_{1-w}^{2/3+w/3} \int_{1-(x_2-w)/2}^{1} \, dx_3 \,dx_2 &&= -5/36 + 7 w/9 - 8 w^2/9
  \end{align}
  
  Multiplying by 2 to account for the ordering $w < x_3 < x_2$:
  
  For $w \in [0, 1/4]$, the win probability is $2(w^2/ 2) = w^2$.
  
  For $w \in [1/4, 1/2]$, the win probability is $2(1/18 - w/9 + w^2/18 -5/36 + 7 w/9 - 8 w^2/9) = -1/6 + 4 w/3 - 5 w^2/3$.
  
  \begin{equation}
    \Pr(w \text{ wins}, w=X_{(1)}, X_{(3)} \text{ elim 1st}) = \begin{cases}
      w^2, & w \in [0 , 1/4]\\
      -1/6 + 4 w/3 - 5 w^2/3, & w \in [1/4, 1/2]
    \end{cases}
      \end{equation}
      
      Plotting:
      
      \begin{center}
  \begin{tikzpicture}
\begin{axis}[xmin=0, xmax=0.5,ymin=0, ymax=0.15, samples=100,xlabel=$w$, ylabel={$\Pr(w \text{ wins}, w = X_{(1)}, X_{(3)} \text{ elim 1st})$}, legend pos=outer north east, legend cell align={left},
 y tick label style={
        /pgf/number format/.cd,
            fixed,
            precision=2,
        /tikz/.cd
    },]

  \addplot[red, thick, name path=b][domain=0:1/4] (x, x^2);
  \addplot[blue, thick, name path=b][domain=1/4:1/2] (x,-1/6 + 4 *x/3 - 5 *x^2/3);

  \legend{$w^2$, $ -1/6 + 4 w/3 - 5 w^2/3$, }
\end{axis}
\end{tikzpicture}
\end{center}
  
  \end{enumerate}
  
  \item $w = X_{(2)}$. Consider the order $x_2 < w < x_3$ (we'll multiply by 2 later to account for $x_3 < w < x_2$).
\begin{enumerate}
\item Candidate 2 is eliminated first.  Since candidate 2 has a smaller vote share than the winner,
  \begin{align*}
    & x_2 + (w - x_2) / 2 < (x_3 - w) / 2 + (w - x_2) / 2\\
    \Leftrightarrow \quad & x_2  < (x_3 - w) / 2\\
    \Leftrightarrow \quad & 2x_2  < x_3 - w\\
    \Leftrightarrow \quad & x_3 > 2x_2 + w.
  \end{align*}
  For this to be feasible, we need 
    \begin{align*}
    & 2x_2 + w < 1\\
    \Leftrightarrow \quad & x_2  < (1-w)/2.
  \end{align*}
  
  Since candidate 2 has a smaller vote share than candidate 3,
    \begin{align*}
    & x_2 + (w - x_2) / 2 < 1-x_3 + (x_3 - w) / 2\\
    \Leftrightarrow \quad & 2x_2 + w - x_2 < 2-2x_3 + x_3 - w \\
    \Leftrightarrow \quad & x_2 + 2w < 2-x_3 \\
    \Leftrightarrow \quad & x_3 < 2 - x_2 - 2w.
  \end{align*}
  For this to be feasible, we need 
  \begin{align*}
    & 2 - x_2 - 2w > w\\
    \Leftrightarrow \quad &  x_2 < 2-3w.
  \end{align*}
  Since $w \le 0.5$, this is always satisfied. The upper bound on $x_3$ is the minimum of $1$ and $2 - x_2 - 2w$. There are equal if 
    \begin{align*}
    & 1 = 2 - x_2 - 2w \\
    \Leftrightarrow \quad &  x_2 = 1 - 2w.
  \end{align*}
To the left of the line $x_2 = 1 - 2w$, the upper bound on $x_3$ is 1; to the right, it's $2 - x_2 - 2w$. 
  
The upper and lower bounds on $x_3$ are equal when
\begin{align*}
  &2x_2 + w = 2 - x_2 - 2w\\
  \Leftrightarrow \quad & 3x_2 = 2 - 3w\\
  \Leftrightarrow \quad & x_2 = 2/3 - w.
\end{align*}
For both constrains to be feasible, we need to be to the left of the line $x_2 = 2/3 - w$. In order for $w$ to win, it needs to be closer to the center than candidate 3. That is, we need $x_3 > 1-w$. This equals the upper bound constraint on $x_3$ if
\begin{align*}
  &1-w = 2 - x_2 - 2w\\
  \Leftrightarrow \quad & x_2 = 1 - w
\end{align*}
Since we already need to be left of the line $2/3 -w$, we don't need to worry about being to the left of $1-w$. Finally, the two lower bounds on $x_3$ are equal if
\begin{align*}
  & 1 -w = 2x_2 + w\\
  \Leftrightarrow \quad & x_2 = 1/2 - w.
\end{align*}
To the left of the line $x_2 = 1/2 - w$, the lower bound on $x_3$ is $1-w$; to the right, the lower bound is $2x_2 + w$.

  \begin{center}
  \begin{tikzpicture}
\begin{axis}[xmin=0, xmax=0.5,ymin=0, ymax=0.5, samples=10,xlabel=$w$, ylabel=$x_2$, legend pos=outer north east, legend cell align={left}]

  \addplot[red, thick, name path=a] (x,x);
  \addplot[blue, thick, name path=b] (x,1/2 - x/2);
    \addplot[green, thick, name path=c] (x,2/3 - x);
  \addplot[violet, thick, name path=d] (x,1/2-x);
  \addplot[orange, thick, name path=d] (x,1-2*x);

  \draw[thick, dotted] (axis cs:1/4, 0) -- (axis cs:1/4, 1/4);
  \draw[thick, dotted] (axis cs:1/3, 1/6) -- (axis cs:1/3, 1/3);
%  \draw[thick, dotted] (axis cs:1/3, 1/3) -- (axis cs:1/3, 2/3);
%  \draw[thick, dotted] (axis cs:2/5, 3/5) -- (axis cs:2/5, 4/5);

   \path[name path=top] (axis cs:0, 1) -- (axis cs:1, 1);
   \path[name path=bottom] (axis cs:0, 0) -- (axis cs:1, 0);

  \node at (axis cs:0.17, 0.08) {A};
  \node at (axis cs:0.32, 0.08) {B};
  \node at (axis cs:0.3, 0.25) {C};
  \node at (axis cs:0.37, 0.19) {D};
  \node at (axis cs:0.46, 0.15) {E};
%  \node at (axis cs:0.38, 0.68) {F};
%  \node at (axis cs:0.45, 0.68) {G};
%

  \addplot[gray!20] fill between[of=a and bottom];
    \addplot[white] fill between[of=c and a];
%    \addplot[white] fill between[of=bottom and d];
%    \addplot[white] fill between[of=d and e];

  \legend{(a) $x_2 = w$, (b) $x_2 = (1-w)/2$, (c) $x_2 = 2/3-w$, (d) $x_2 = 1/2 - w$, (e) $x_2 = 1 - 2w$}
\end{axis}
\end{tikzpicture}
\end{center}
Lines (a) and (d) intersect at $w=1/4$; lines (a) and (c) intersect at $w=1/3$. Integrating over the regions:

\begin{align*}
  A: \quad & \int_{0}^{w} \int_{1-w}^{1} \, dx_3 \,dx_2 &&= w^2\\
  B: \quad &  \int_{0}^{1/2-w} \int_{1-w}^{1} \, dx_3 \,dx_2 &&= w/2 - w^2\\
  C: \quad &  \int_{1/2-w}^{w} \int_{2x_2+w}^{1} \, dx_3 \,dx_2 &&= -1/4 + 3 w/2 - 2 w^2\\
  D: \quad &  \int_{1/2-w}^{1-2w} \int_{2x_2+w}^{1} \, dx_3 \,dx_2 &&= -1/4 + 3 w/2 - 2 w^2\\
  E: \quad &  \int_{1-2w}^{2/3-w} \int_{2x_2+w}^{2-x_2-2w} \, dx_3 \,dx_2 &&= 1/6 - w + 3 w^2/2 \end{align*}

We now sum and multiply by 2 to account for the ordering $x_3 < w < x_2$.

For $x \in [0, 1/4]$, the win probability is $2w^2$.

For $x \in [1/4, 1/3]$, the win probability is $2(w/2 - w^2 -1/4 + 3 w/2 - 2 w^2) = -1/2 + 4 w - 6 w^2$.

For $x \in [1/3, 1/2]$, the win probability is $2(w/2 - w^2 -1/4 + 3 w/2 - 2 w^2 + 1/6 - w + 3 w^2/2) = -1/6 + 2 w - 3 w^2$.

Summarizing and visualizing:

\begin{equation}
  \Pr(w \text{ wins}, w = X_{(2)}, X_{(1)} \text{ elim 1st}) = \begin{cases}
    2w^2, & w\in [0, 1/4]\\
    -1/2 + 4 w - 6 w^2, & w \in [1/4, 1/3]\\
    -1/6 + 2 w - 3 w^2, & w \in [1/3, 1/2]
  \end{cases}
\end{equation}

\begin{center}
  \begin{tikzpicture}
\begin{axis}[xmin=0, xmax=0.5,ymin=0, ymax=0.2, samples=100,xlabel=$w$, ylabel={$\Pr(w \text{ wins}, w = X_{(2)})$}, legend pos=outer north east, legend cell align={left},
 y tick label style={
        /pgf/number format/.cd,
            fixed,
            precision=2,
        /tikz/.cd
    },]

  \addplot[red, thick, name path=b][domain=0:1/4] (x,2*x^2);
  \addplot[blue, thick, name path=b][domain=1/4:1/3] (x,-1/2 + 4 *x - 6 *x^2);
  \addplot[green, thick, name path=b][domain=1/3:1/2] (x,-1/6 + 2 *x - 3 *x^2);
%  \addplot[violet, thick, name path=b][domain=1/3:2/5] (x,-1 + 6* x - 15 *x^2/2);
%  \addplot[orange, thick, name path=b][domain=2/5:1/2] (x,-1/3 + 8* x/3 - 10 *x^2/3);

  \legend{$2w^2$, $-1/2 + 4 w - 6 w^2$, $-1/6 + 2 w - 3 w^2$}
\end{axis}
\end{tikzpicture}
\end{center}

\item Candidate 3 is eliminated first. Since candidate 3 has a smaller vote share than the winner, 
\begin{align*}
  &1 - x_3 + (x_3 -w) / 2 < (x_3 - w) / 2 + (w - x_2) / 2\\
  \Leftrightarrow \quad & 1 - x_3 < (w - x_2) / 2\\
  \Leftrightarrow \quad & x_3 > 1- (w - x_2) / 2
\end{align*}
This is feasible if 
\begin{align*}
  & 1- (w - x_2) / 2 < 1\\
  \Leftrightarrow \quad & - w + x_2 < 0\\
   \Leftrightarrow \quad &  x_2 < w,
\end{align*}
which is always true. Since candidate 3 has a smaller vote share than candidate 2,
\begin{align*}
  & 1 - x_3 + (x_3 -w) / 2 < x_2 + (w - x_2) / 2\\
   \Leftrightarrow \quad & 2 - 2x_3 + x_3 -w < 2x_2 + w - x_2\\
   \Leftrightarrow \quad & 2 - x_3 < x_2 + 2w\\
   \Leftrightarrow \quad & x_3 > 2 - x_2 - 2w 
\end{align*}
This is feasible if 
\begin{align*}
  &2 - x_2 - 2w < 1\\
  \Leftrightarrow \quad & x_2 > 1 - 2w.
\end{align*}
The two lower bounds on $x_3$ are equal if
\begin{align*}
  &1- (w - x_2) / 2 = 2 - x_2 - 2w \\
  \Leftrightarrow \quad & 2- w + x_2 = 4 - 2x_2 - 4w\\
  \Leftrightarrow \quad & 3x_2 = 2 - 3w\\
  \Leftrightarrow \quad & x_2 = 2/3 - w
\end{align*}
Above the line $2/3 - w$, the lower bound on $x_3$ is $1- (w - x_2) / 2$; below the line, it's $2 - x_2 - 2w$. As long as candidate 3 is eliminated first, $w$ wins since it's closer to the center than $x_2$. 

  \begin{center}
  \begin{tikzpicture}
\begin{axis}[xmin=0, xmax=0.5,ymin=0, ymax=0.5, samples=10,xlabel=$w$, ylabel=$x_2$, legend pos=outer north east, legend cell align={left}]

  \addplot[red, thick, name path=a] (x,x);
  \addplot[blue, thick, name path=b] (x,1 - 2*x);
    \addplot[green, thick, name path=c] (x,2/3 - x);
%  \addplot[violet, thick, name path=d] (x,1/2-x);
%  \addplot[orange, thick, name path=d] (x,1-2*x);

%  \draw[thick, dotted] (axis cs:1/4, 0) -- (axis cs:1/4, 1/4);
%  \draw[thick, dotted] (axis cs:1/3, 1/6) -- (axis cs:1/3, 1/3);
%  \draw[thick, dotted] (axis cs:1/3, 1/3) -- (axis cs:1/3, 2/3);
%  \draw[thick, dotted] (axis cs:2/5, 3/5) -- (axis cs:2/5, 4/5);

   \path[name path=top] (axis cs:0, 1) -- (axis cs:1, 1);
   \path[name path=bottom] (axis cs:0, 0) -- (axis cs:1, 0);

  \node at (axis cs:0.45, 0.17) {A};
  \node at (axis cs:0.45, 0.33) {B};
%  \node at (axis cs:0.3, 0.25) {C};
%  \node at (axis cs:0.37, 0.19) {D};
%  \node at (axis cs:0.46, 0.15) {E};
%  \node at (axis cs:0.38, 0.68) {F};
%  \node at (axis cs:0.45, 0.68) {G};
%

  \addplot[gray!20] fill between[of=a and bottom];
    \addplot[white] fill between[of=b and bottom];
%    \addplot[white] fill between[of=bottom and d];
%    \addplot[white] fill between[of=d and e];

  \legend{(a) $x_2 = w$, (b) $x_2 = 1 - 2w$, (c) $x_2 = 2/3 - w$}
\end{axis}
\end{tikzpicture}
\end{center}
Lines (a), (b), and (c) all intersect at $w = 1/3$. Integrating over the two regions and multiplying by 2 to account for the ordering $x_3 < w < x_2$:

\begin{align*}
  A: \quad & \int_{1-2w}^{2/3-w} \int_{2-x_2-2w}^{1} \, dx_3 \,dx_2 &&= 1/18 - w/3 + w^2/2\\
  B: \quad &  \int_{2/3-w}^{w} \int_{1-(w-x_2)/2}^{1} \, dx_3 \,dx_2 &&= 1/9 - 2 w/3 + w^2\\
 \end{align*}

For $w \in [1/3, 1/2]$, the win probability is $2(1/18 - w/3 + w^2/2 + 1/9 - 2 w/3 + w^2) = 1/3 - 2 w + 3 w^2$. Thus:
\begin{equation}
  \Pr(w \text{ wins}, w = X_{(2)}, X_{(3)} \text{ elim 1st}) = 1/3 - 2 w + 3 w^2, \qquad w \in [1/3, 1/2]
\end{equation}

\end{enumerate}
 \item $w = X_{(3)}$. If both $x_2 < w$ and $x_3 < w$, then $w$ wins by IRV. Thus, $\Pr(w \text{ wins}, w = X_{(3)}) = w^2$ for $w \in [0, 0.5]$. 
\end{enumerate}

We can finally sum over the three cases to arrive at $\Pr(w \text{ wins})$. 

For $w \in [0, 1/6]$, the sum is $w^2 + 2w^2 + w^2 = 4w^2$.

For $w \in [1/6, 1/4]$, the sum is $w^2 + 2w^2 + w^2 + 1/3 - 4w + 12w^2 = 1/3 - 4 w + 16 w^2$.

For $w \in [1/4, 1/3]$, the sum is $w^2 - 1/2 + 4w - 6w^2 - 1/6 + 4w/3 - 5w^2/3 -1 + 20w/3 - 28w^2/3 = -5/3 + 12 w - 16 w^2$.

For $w \in [1/3, 1/2]$, the sum is $w^2 + 1/3 - 2 w + 3 w^2 - 1/6 + 2w - 3w^2 - 1/6 + 4w/3 - 5w^2/3 - 1/3 + 8w/3 - 10w^2/3 = -1/3 + 4 w - 4 w^2$.

Summarizing and plotting:

\begin{equation}
  \Pr(w \text{ wins}) = \begin{cases}
  4w^2, & w \in [0, 1/6]\\
  1/3 - 4 w + 16 w^2, & w \in [1/6, 1/4]\\
  -5/3 + 12 w - 16 w^2, & w \in [1/4, 1/3]\\
  -1/3 + 4 w - 4 w^2, & w \in [1/3, 1/2]
   \end{cases}
\end{equation} 

\begin{center}
  \begin{tikzpicture}
\begin{axis}[xmin=0, xmax=0.5,ymin=0, ymax=0.7, samples=100,xlabel=$w$, ylabel={$\Pr(w \text{ wins})$}, legend pos=outer north east, legend cell align={left},
 y tick label style={
        /pgf/number format/.cd,
            fixed,
            precision=1,
        /tikz/.cd
    },]

  \addplot[red, thick, name path=b][domain=0:1/6] (x, 4*x^2);
  
  \addplot[blue, thick, name path=b][domain=1/6:1/4] (x, 1/3 - 4 *x + 16 *x^2);
  \addplot[green, thick, name path=b][domain=1/4:1/3] (x, -5/3 + 12 *x - 16 *x^2);
  \addplot[violet, thick, name path=b][domain=1/3:1/2] (x, -1/3 + 4 *x - 4 *x^2);
 
  \legend{$4w^2$, $1/3 - 4 w + 16 w^2$, $-5/3 + 12 w - 16 w^2$, $-1/3 + 4 w - 4 w^2$}
\end{axis}
\end{tikzpicture}
\end{center}  

To get the winner position distribution $f_{R_3}$, we scale by three. 
  The variance of $f_{R_3}$ is thus:
  \begin{equation}
  \begin{split}
    6\Bigg[&\int_{0}^{1/6} (w - 1/2)^2(4w^2) \, dw + \int_{1/6}^{1/4} (w - 1/2)^2(1/3 - 4 w + 16 w^2) \, dw \\
     &+ \int_{1/4}^{1/3} (w - 1/2)^2(-5/3 + 12 w - 16 w^2) \, dw + \int_{1/3}^{1/2} (w - 1/2)^2(-1/3 + 4 w - 4 w^2) \, dw \Bigg] = 25/864 \approx 0.029
  \end{split}
   \end{equation}
\clearpage

\section{Additional Figures}\label{app:figures}
\begin{figure}[h]
\centering
  \includegraphics[width=\textwidth]{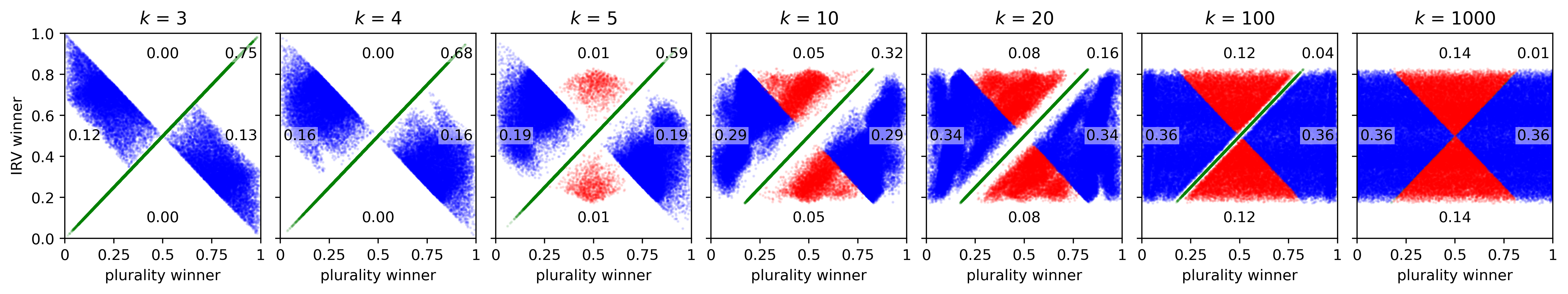}
  \caption{Plurality vs.\ IRV winner positions in 100,000 simulation trials for increasing candidate count $k$ (with uniform voters and candidates). Blue points are trials where the IRV winner was more moderate than the plurality winner, while red points are trials where the plurality winner was more moderate. Green points are trials where the winners were identical. Numbers in each quadrant show the proportion of trials falling in that region (the top right number is the proportion of same-winner trials). Notice that cases where the IRV winner is more extreme only appear beginning at $k=5$, in accordance with \Cref{thm:small-k-not-more-extreme}. Note the probabilistic moderating effect of IRV compared to plurality: IRV does not elect extreme candidates as $k$ grows large, but plurality does.}\label{fig:plurality-irv-scatter}
\end{figure}

\begin{figure}[h]
  \centering
  \includegraphics[width=0.4\textwidth]{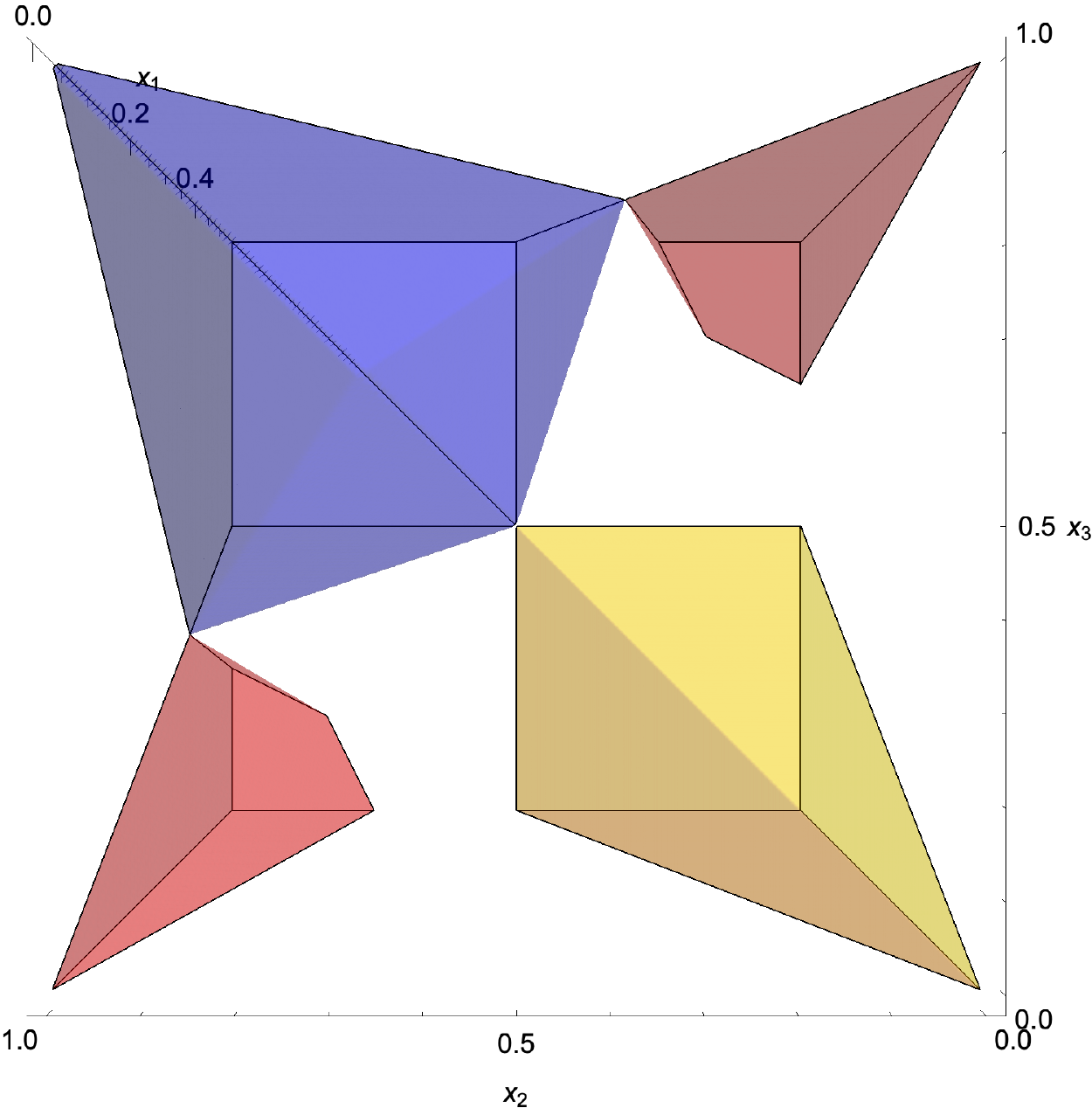}
  \qquad
   \includegraphics[width=0.4\textwidth]{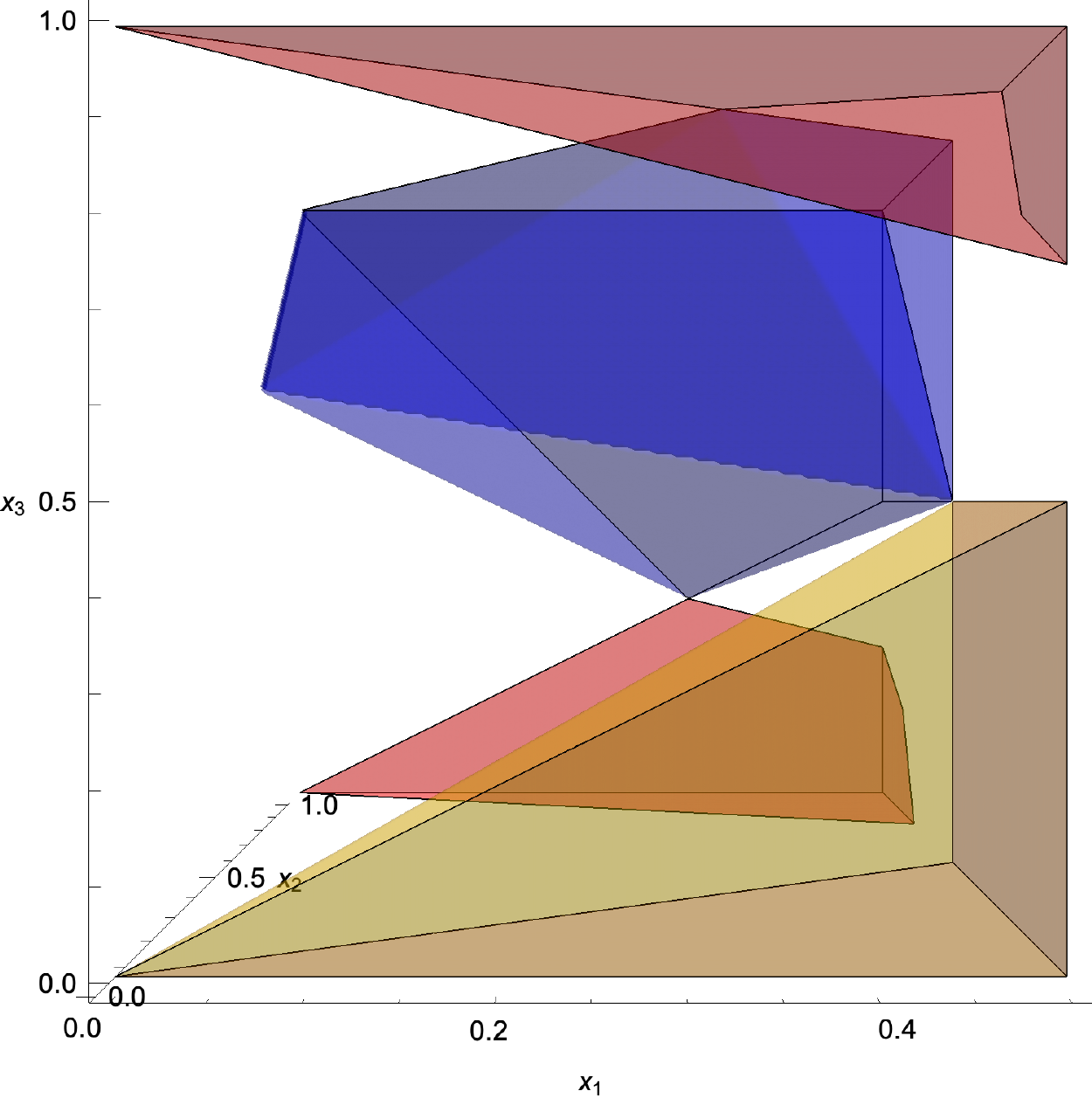}
   \caption{Plurality winner regions for $k=3$ with uniform voters and candidates. Colored polyhedra show the regions where a candidate at position $x_1$ is the plurality winner against candidates at $x_2$ and $x_3$. Regions are only shown for $x_1 \le 0.5$, since the other half of is symmetric. The color of a region corresponds to the order statistic of the winner. Blue: winner is the leftmost, red: winner is in the middle, yellow: winner is the rightmost. The left view has the plane of the page at $x_1=0$, looking towards increasing $x_1$. The right view has the plane of the page at $x_2 = 0$, with $x_1$ increasing from left to right. }\label{fig:plurality-regions}
\end{figure}
\begin{figure}[h]
  \centering
  \includegraphics[width=0.4\textwidth]{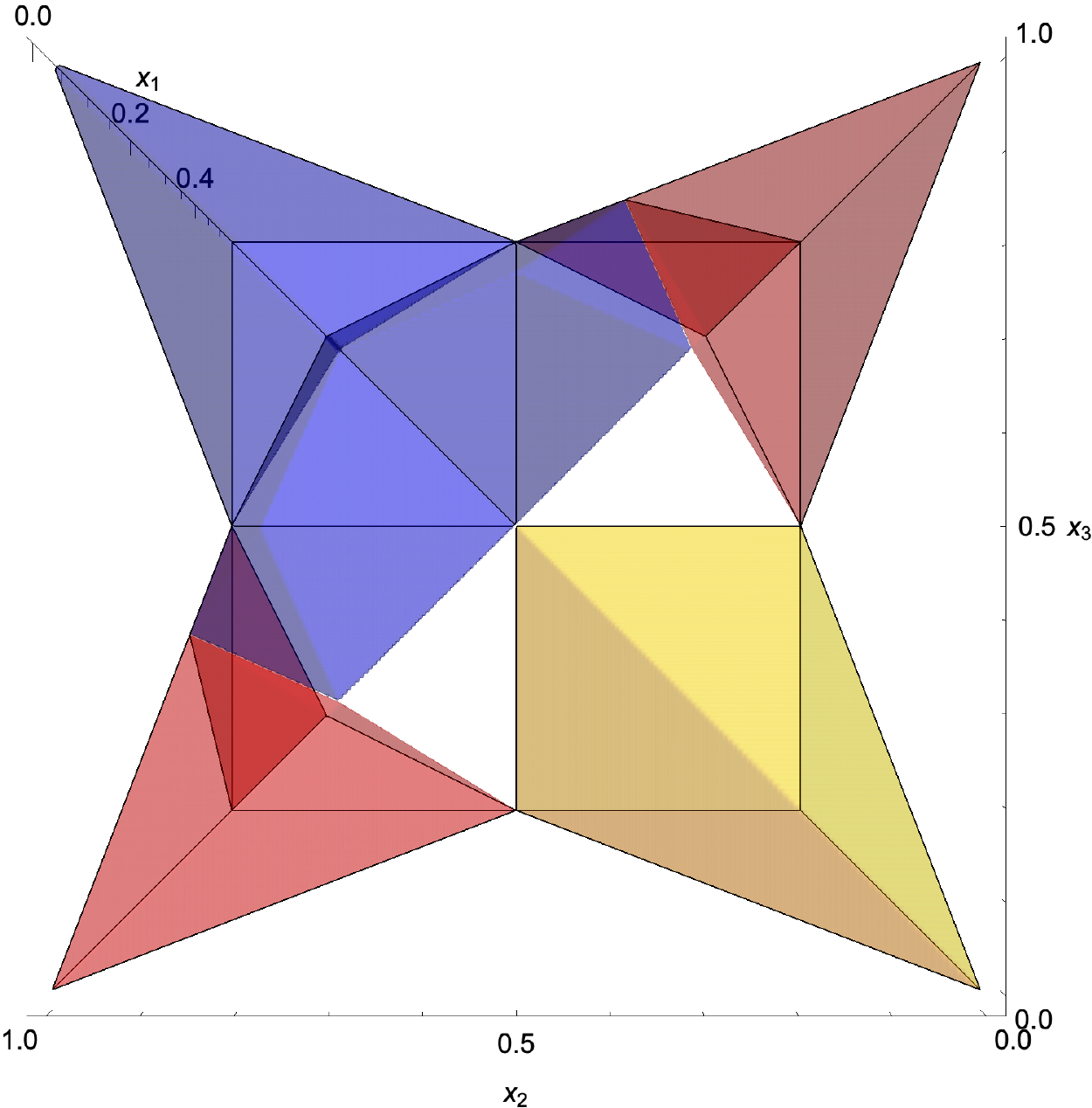}
  \qquad
   \includegraphics[width=0.4\textwidth]{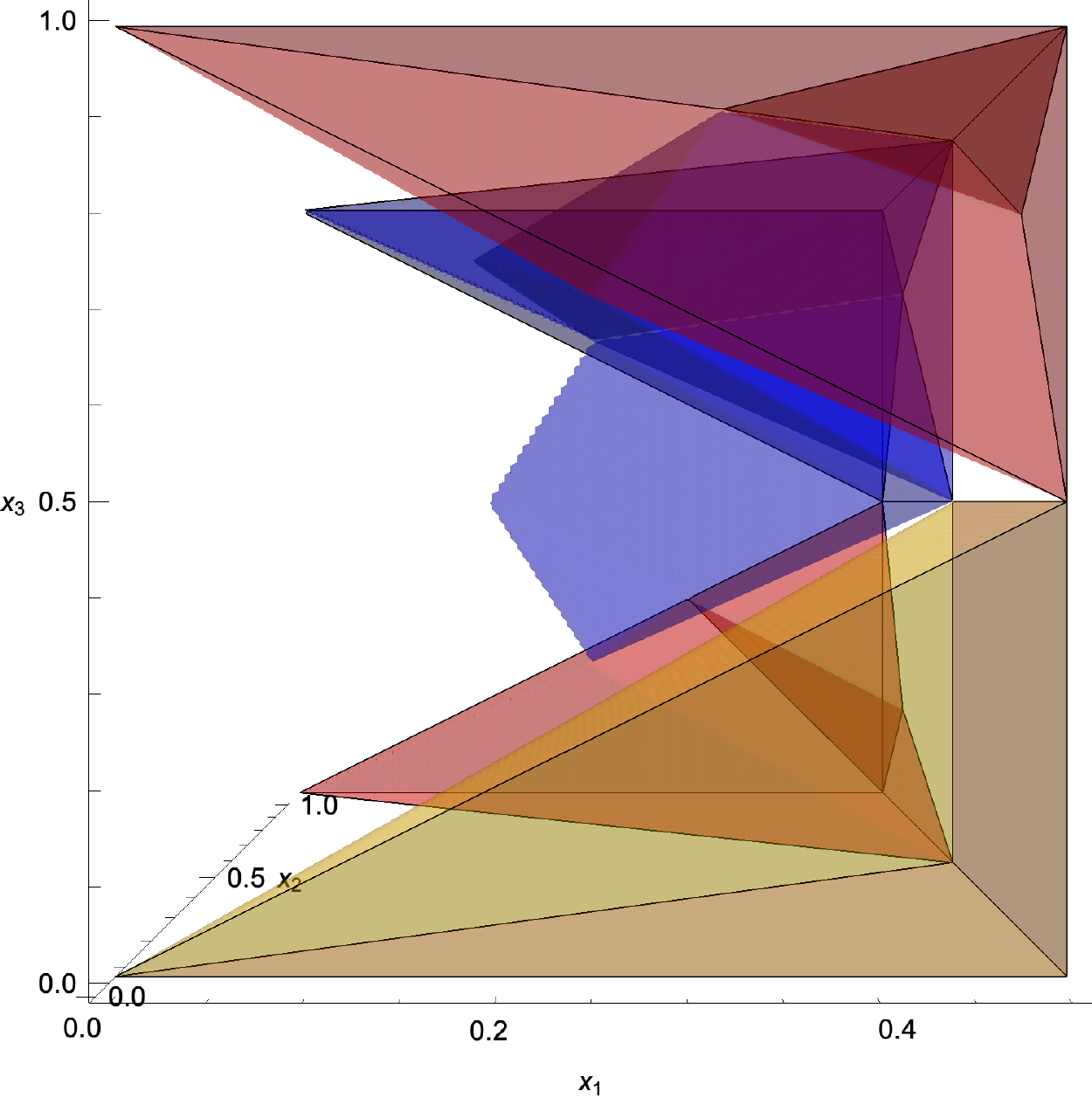}
   \caption{IRV winner regions for $k=3$ with uniform voters and candidates. See \Cref{fig:plurality-regions} for details about the visualization.}\label{fig:irv-regions}
\end{figure}

\end{document}